\documentclass{article} 
\usepackage{iclr2027_conference,times}

\usepackage{amsmath,amsfonts,bm}

\def\eqref#1{equation~\ref{#1}}

\def\1{\bm{1}}

\DeclareMathAlphabet{\mathsfit}{\encodingdefault}{\sfdefault}{m}{sl}
\SetMathAlphabet{\mathsfit}{bold}{\encodingdefault}{\sfdefault}{bx}{n}

\newcommand{\R}{\mathbb{R}}

\usepackage{hyperref}
\usepackage{url}
\usepackage{amsthm}
\usepackage{algorithm}
\usepackage{nicefrac}
\usepackage[noend]{algpseudocode}
\usepackage{mathtools}
\usepackage{amssymb}

\theoremstyle{plain}
\newtheorem{theorem}{Theorem}[section]
\theoremstyle{definition}
\newtheorem{definition}[theorem]{Definition}

\newtheorem{lemma}[theorem]{Lemma}
\newtheorem{corollary}[theorem]{Corollary}

\newtheorem{fact}[theorem]{Fact}

\newtheorem*{remark}{Remark}

\newcommand{\bmnote}[1]{\textcolor{red}{\bf{#1}}}

\newcommand{\remove}[1]{}

\title{Differentially Private Approximation of the John Ellipsoid}

\author{Bar Mahpud \\
Faculty of Engineering \\
Bar-Ilan University \\
Israel \\
\texttt{mahpudb@biu.ac.il} \\
\And
Or Sheffet \\
Faculty of Engineering \\
Bar-Ilan University \\
Israel \\
\texttt{or.sheffet@biu.ac.il} \\
\And
Daniel Omer \\
Einstein Institute of Mathematics \\
The Hebrew University of Jerusalem \\
Israel \\
\texttt{daniel.omer@mail.huji.ac.il}
}

\iclrpreprint
\begin{document}

\maketitle

\begin{abstract}
We study the problem of approximating the John ellipsoid (JE) of a given (centrally symmetric) polytope of $n$ constraints in a Euclidean space under differential privacy (DP). We give the first differentially private algorithm for this problem under the standard model, where neighboring datasets may differ arbitrarily in one a single constraint. Our work also extends to the complimentary problem of Minimum Enclosing Ellipsoid of $n$ points in the Euclidean space.

Our approach is based on the recent non-private multiplicative-weights algorithm of~\cite{pmlr-v99-cohen19a}. First we introduce a non-private generalization of the Cohen et al algorithm, yielding a $(1+\gamma)$-approximation of the JE problem while violating at most $\kappa n$ constraints in $O(\log(1/\kappa)/\gamma)$ iterations. This variant works by projecting the intermediate weights assigned to the constraints onto the set of $\kappa$-dense distributions, similarly to~\cite{bun2020efficientnoisetolerantprivatelearning}.

We then design a $\rho$-zCDP variant of this algorithm by adding Gaussian noise to the weighted covariance matrix aggregated in each step of the algorithm. 
Under a mild goodness assumption on the data we can assert that the resulting noisy matrix is close to the true matrix, thereby achieving essentially the same guarantee as the non-private algorithm provided sufficiently many input points. Thus our method achieves an efficient DP poly-time algorithm under concrete sample complexity bounds.
\end{abstract}

\section{Introduction}
\label{sec:intro}
The John ellipsoid is one of the most fundamental geometric objects associated with a convex body. Given a convex polytope $K \subseteq \mathbb{R}^d$, the John Ellipsoid (JE) is defined as the maximum-volume ellipsoid contained in $K$. A classical theorem of~\cite{john1948extremum} shows that this ellipsoid provides a canonical geometric approximation of the body: if $E \subseteq K$ is the John ellipsoid, then $K \subseteq d \cdot E$, where $d \cdot E$ denotes a dilation of $E$ by a factor of $d$ about its center. Moreover, when $K$ is centrally symmetric (i.e. $\forall x \in K$ its antipode $-x$ is also in $K$) and the JE must be centered at the origin, this guarantee improves to $K \subseteq \sqrt{d}\, E$.

These structural properties make the John ellipsoid a central tool in high-dimensional geometry and its algorithmic applications. These include sampling and numerical integration~\citep{Vempala_Mathematical_2005, cdwy18}, linear bandit problems~\citep{bcbk12, JMLR:v17:hazan16a}, linear programming~\citep{ls15}, and differential privacy bounds~\citep{ntz13}.
Beyond these areas, the John ellipsoid also plays an important role in experimental design, a classical topic in statistics~\citep{atw69}. In particular, the D-optimal design problem seeks to maximize the determinant of the Fisher information matrix~\citep{Kiefer_Wolfowitz_1960, atw69}, which is known to be equivalent to computing the John ellipsoid of a symmetric polytope (see, e.g.,~\cite{tod16}). D-optimal design has attracted renewed interest in the machine learning community~\citep{pmlr-v70-allen-zhu17e, JMLR:v18:17-175, lfn18}.

As is common, we consider in this work symmetric polytopes of the form $K = \{x \in \mathbb{R}^d : -\mathbf{1} \le Px \le \mathbf{1}\}$, where the $n\times d$ matrix $P$ is composed of $n$ constraints each denote by $x_i$.
In many applications, these $n$ constraints may contain sensitive information—such as locations, personal attributes, or high-dimensional feature vectors—motivating the need for private algorithms. Differential privacy~\citep{dwork2006calibrating} requires that the output of an algorithm is stable under the modification of a single $x_i$ in $P$, thereby protecting individual-level information. However, this requirement is particularly challenging in geometric settings: core quantities such as covariance matrices and enclosing or inscribed ellipsoids can be highly sensitive, and in some instances a single point may significantly alter the optimal solution. This is particularly true for the JE problem, where by adding to $P$ a new constraint $x_{n+1}$ we may obtain a new polytope $K'$ of significantly smaller volume than $K$, and therefore a significantly smaller JE.
%
%
This inherent instability makes the direct computation of the John ellipsoid incompatible with privacy constraints. Consequently, we adopt a \emph{bi-criteria formulation}: we allow a multiplicative approximation of $1+\gamma$ to the JE volume objective and we also allow a relaxation of the containment requirement~-- it may violate up to a $\kappa$-fraction of the constraints.

\paragraph{Our contribution.}
In this work, we present a new framework for approximating the John ellipsoid under differential privacy, based on multiplicative-weights algorithm and KL-projection. Our starting point is the non-private algorithm of~\cite{pmlr-v99-cohen19a}, which uses a multiplicative update over weights and returns, after \(O(\log(n/d)/\gamma)\) iterations, a $(1+\gamma)$-approximation of the input's JE. Our design of a DP-equivalent of this algorithm is composed of two stages.

\medskip

\emph{The first step} is to introduce a (non-private) generalization of the algorithm of Cohen et al, in which the weights assigned to constraints in each iteration must be $\kappa$-dense~--- namely, each weight is upper bounded by $\frac 1 \kappa \cdot \frac d n$. We do so by projecting the true weights on the set of $\kappa$-dense weights, a projection that is simple to compute. In Section~\ref{sec:non-private-alg} we analyze the resulting algorithm and present the following guarantee. 

\medskip

\begin{theorem}
Let $P=\{x_1,\dots,x_n\}\subseteq\mathbb{R}^d$.
For any $\gamma,\kappa\in(0,1)$, after $T = O\!\left(\frac{\log (1/\kappa)}{\gamma}\right)$ iterations, Algorithm~\ref{alg:non-private} outputs a positive definite matrix \(M\), defining the ellipsoid
\[
E
:=
\{x\in\mathbb R^d:x^\top M^{-1}x\le1\},
\]
such that: 
    there exists a subset $S \subseteq [n]$ with $|S|\ge (1-\kappa)n$ such that
    \[
    e^{-\gamma/2}E \subseteq K_S,
    \]
    where $K_S := \{x\in \R^d : |x_i^\top x|\le 1 \ \forall i\in S\}$; and, for the original polytope we have
    \[
    K \subseteq \sqrt d\, E,
    \]
    where $K:=\{x\in\mathbb R^d: |x_i^\top x|\le 1 \text{ for all } i\in[n]\}$;
    and moreover, we have
    \[
    \operatorname{vol}(e^{-\gamma/2}E)
    \;\ge\;
    e^{-d\gamma/2}\operatorname{vol}(E^*),
    \]
    where $E^*$ is the John ellipsoid of $K$.
\end{theorem}
Our ellipsoid is $\gamma$-close to the John Ellipsoid in any direction, 
(just like as the approximation of Cohen et al). If we
wish to get $\gamma^*$-close approximation in volume, we simply set $\gamma = \gamma^*/d$, resulting in an algorithm that runs in $O\!\left(\frac{d\log (1/\kappa)}{\gamma^*}\right)$ many iterations.


\medskip

\emph{The second step} is to present a zCDP analogue of the above-mentioned algorithm. We do simply by applying the Gaussian mechanism to the weighted covariance matrix aggregated in each step of the algorithm. Leveraging on the fact that the weight of each $x_i$ is upper bounded due to $\kappa$-density, it is easy to upper bound the sensitivity of the weighted covariance matrix and scale the Gaussian noise accordingly. However, in order to assert that the added noise does not interfere with the weights assigned to the constraints we need to rely on a goodness assumption.

\begin{definition}[$(\kappa,\tau)$-good input]\label{def:good-input}
	We say that $P=\{x_1,\ldots,x_n\}\subseteq\mathbb R^d$ is
	$(\kappa,\tau)$-good if for every unit vector $v\in\mathbb S^{d-1}$, defining
	\[
	\mathrm{BAD}_v
	:=
	\left\{
	i\in[n] : \langle x_i,v\rangle^2 < \frac{2\tau}{d}
	\right\},
	\]
	we have $|\mathrm{BAD}_v| < \frac{\kappa n}{2}$.
\end{definition}
For example, when the $n$ constraints are drawn from the uniform distribution over the unit-sphere, then the goodness assumption holds for constant values of $\kappa$ and $\tau$.
\ificlrfinal
\else
  \ificlrpreprint
  \else
    \pagebreak
  \fi
\fi
For our purposes, the \((\kappa,\tau)\)-goodness assumption 
ensures that every $\kappa$-dense measure induces a well-conditioned covariance matrix, which is needed to control the effect of Gaussian noise in the private algorithm (See Lemma~\ref{lem:good-assumption}). 
Altogether, in Section~\ref{sec:private-alg} we introduce and analyze our $\rho$-zCDP algorithm, obtaining the following guarantee.


\begin{theorem}
Let $P=\{x_1,\dots,x_n\}\subseteq\mathbb{R}^d$ be $(\kappa, \tau)$-good input with range $R$.
For any $\rho$ and $\gamma, \beta\in(0,1)$, Algorithm~\ref{alg:private-alg} is $\rho$-zCDP, and, after $T = O\!\left(\frac{\log (1/\kappa)}{\gamma}\right)$ iterations, supposing that
\begin{align*}
n \ge \frac{32dR^2(\sqrt{d}+\sqrt{2\log(\frac{T+1}{\beta})})}{\kappa\gamma\tau\sqrt{2\rho/(T+1)}} = O\left(\frac{dR^2(\sqrt{d}+\sqrt{\log(\frac{\log (1/\kappa)}{\gamma\beta})})\sqrt{\log (1/\kappa)}}{\kappa\gamma^{1.5}\tau\sqrt{\rho}}\right)
\end{align*}
it outputs with probability $\ge 1-\beta$ a positive definite matrix \(M\), defining the ellipsoid
\[
E
:=
\{x\in\mathbb R^d:x^\top M^{-1}x\le1\},
\]
such that:
\begin{itemize}   
    \item there exists a subset $S \subseteq [n]$ with $|S|\ge (1-\kappa)n$ such that
    \[
    \sqrt{(1-\gamma/8)}e^{-\gamma/2}E \subseteq K_S,
    \]
    where $K_S := \{x\in \R^d : |x_i^\top x|\le 1 \ \forall i\in S\}$; and, for the original polytope we have
    \[
    K \subseteq \sqrt{d(1+\gamma/8)}\, E,
    \]
    where $K:=\{x\in\mathbb R^d: |x_i^\top x|\le 1 \text{ for all } i\in[n]\}$.
    
    \item moreover, we have
    \[
    \operatorname{vol}(\sqrt{(1-\gamma/8)}e^{-\gamma/2}E)
    \;\ge\;
    e^{-9d\gamma/14}\operatorname{vol}(E^*),
    \]
    where $E^*$ is the John ellipsoid of $K$.
\end{itemize}
\end{theorem}

We comment that our algorithm 
does not require
the goodness condition to be checked before it is executed. 
The goodness assumption is used only in our utility analysis, guaranteeing that each weighted second-moment matrix the algorithm iterates through has sufficiently large least eigenvalue. 
Accordingly, our sample-complexity bound should be interpreted as a
sufficient guarantee for inputs satisfying this condition, rather than as a
characterization of all inputs on which the algorithm succeeds ~--- as our algorithm may succeed even on inputs that are not $(\kappa,\tau)$-good.

\paragraph{Extensions.}
Our problem also admits an equivalent interpretation through polar duality.
Let $C_P:=\operatorname{conv}\{\pm x_1,\ldots,\pm x_n\}$. Then the centrally symmetric polytope considered throughout this work satisfies
\(K_P=C_P^\circ\), where \((\cdot)^\circ\) denotes the polar body. Since
polarity reverses containment and maps origin-centered ellipsoids to
origin-centered ellipsoids, the John ellipsoid of \(K_P\) is precisely the
polar of the minimum-volume enclosing ellipsoid (MVEE) of \(C_P\).
Consequently, our private algorithm also yields guarantees for a robust
version of the centered MVEE problem in which a \(\kappa\)-fraction of the
input points may be discarded as outliers.

This MVEE interpretation further allows us to extend our results beyond
centrally symmetric, origin-centered ellipsoids. In particular, the standard
homogeneous-coordinate lifting for the MVEE problem~\citep{tod16} embeds each
point \(y_i\in\mathbb{R}^d\) as \(x_i=(y_i,1)\in\mathbb{R}^{d+1}\), thereby reducing the uncentered MVEE problem in \(\mathbb{R}^d\) to a centered
problem in one higher dimension. We show that this reduction extends to our
trimmed setting and preserves both the approximation factor and the trimming
guarantee. Hence, through the MVEE interpretation above, our results extend
also to general uncentered ellipsoids. For the private guarantee in the
uncentered setting, we additionally require that the lifted points satisfy
the corresponding goodness assumption. We discuss this in Section~\ref{sec:extensions}.

\paragraph{Related Work.}
Computing or approximating the John Ellipsoid (JE) has been extensively studied in optimization, computational geometry, and statistics. Classical approaches rely on reductions to the minimum volume enclosing ellipsoid (MVEE) and interior-point methods~\citep{KT93, NN94}, with subsequent improvements achieving a $(1+\gamma)$-approximation in time $O(dn^3/\gamma)$~\citep{KY05, TY07}. The problem is also equivalent to $D$-optimal design~\citep{atw69, tod16}, which has motivated first-order and proximal methods~\citep{lfn18, GP18}. More recently, faster algorithms based on iterative updates and sketching techniques have been developed~\citep{pmlr-v99-cohen19a, NEURIPS2024_7e1854dc, cao2026faster, li2024quantumspeedupsapproximatingjohn}, with the method of~\cite{pmlr-v99-cohen19a} achieving near-optimal iteration complexity via a simple multiplicative update.

In the context of differential privacy, much less is known for geometric optimization problems of this type. A closely related line of work studies the private \(1\)-cluster problem, introduced by~\cite{NissimStemmerVadhan16}, where the goal is to find a small-radius ball containing almost all of a prescribed number of input points. Subsequent work obtained improved guarantees for this problem in both the central and local models~\citep{NissimStemmer18}, and~\cite{GhaziKumarManurangsi20} gave efficient private algorithms achieving essentially the corresponding non-private approximation ratios, up to a small additive loss in the number of covered points. For the closely related Minimum Enclosing Ball (MEB) problem,~\cite{MahpudS22} gave the first linear-time differentially private fPTAS, achieving a \((1+\gamma)\)-approximation to the optimal radius while allowing a small fraction of points to remain uncovered. Taken together, these works illustrate a recurring feature of private geometric optimization: under the standard neighboring-dataset model, useful utility guarantees naturally take a bi-criteria form, simultaneously approximating the geometric objective and permitting a small number of constraints or points to be violated.


\cite{Gu2024FastJE} proposed a differentially private algorithm for approximating the John ellipsoid under a weaker notion of neighboring inputs, where a single datum can be perturbed only by a small $\epsilon_0$ (i.e., the two neighboring instances differ on a single entry where $\|x_i-x'_i\|\leq \epsilon_0$). Their claim-to-fame is to upper-bound the $\ell_2$-sensitivity of the weighted leverage-score vector by $L\epsilon_0$ where $L$ is some factor left unspecified\footnote{They bound $L$ as some (unspecified) polynomial in the numerous parameters introduced by their various assumptions.}; thus allowing them to revise the algorithm of Cohen et al. so that \emph{all} $n$ constraints are approximately satisfied (a guarantee clearly infeasible under standard DP). We, in contrast, use the the standard definition of DP while introducing concrete privacy-utility trade-off in the form of an explicit sample complexity bound (and also achieving smaller iteration complexity then their algorithm).

Our use of dense measures is motivated by the
smooth boosting literature, particularly the lazy Bregman-projection framework
of~\cite{bun2020efficientnoisetolerantprivatelearning}, building on earlier work of
\cite{barak2009uniform}. In that line of work, multiplicative weights are
combined with projection onto a dense class of measures to keep the influence
of every datum controlled in a similar fashion to \cite{Roth2014}. We adapt this mechanism to the John ellipsoid
problem, where the maintained measure determines a weighted covariance matrix
and the multiplicative update is driven by logarithmic leverage-type scores.

\section{Preliminaries}
\label{sec:preliminaries}

\paragraph{John ellipsoid and geometric formulation.}
Let $P=\{x_1,\ldots,x_n\}\subseteq\mathbb{R}^d$,
where $\operatorname{span}(P)=\mathbb{R}^d$.
We consider the centrally symmetric polytope
\[
K_P :=
\left\{
x\in\mathbb{R}^d:
|x_i^\top x|\le 1,\ \forall i\in[n]
\right\}
\]
The John ellipsoid $E^*$ of $K_P$ is the
maximum-volume ellipsoid contained in $K_P$.
By symmetry, $E^*$ is centered at the origin.

For a nonnegative measure
$\mu\in\mathbb{R}_{\ge 0}^n$
with $\sum_{i=1}^n\mu_i=d$, define
\[
\Sigma(\mu):=\sum_{i=1}^n\mu_i x_ix_i^\top,
\qquad
M(\mu):=\Sigma(\mu)^{-1},
\]
whenever $\Sigma(\mu)$ is invertible.
The corresponding ellipsoid is
\[
E(\mu):=
\left\{
x\in\mathbb{R}^d:
x^\top M(\mu)^{-1}x\le 1
\right\}
\]
and its volume is proportional to
$\det(M(\mu))^{1/2}$.

The primal-dual formulation of the John ellipsoid
provides a geometric interpretation of these
measures and the corresponding volume guarantees.
We recall this formulation and its consequences
in Appendix~\ref{apx:additional-background}. 

\begin{definition}
[$(1+\gamma)$-approximate $\kappa$-trimmed John solution]
Let $\gamma>0$ and $\kappa\in(0,1)$.
A nonnegative measure $\mu$ of total mass $d$
is a $(1+\gamma)$-approximate
$\kappa$-trimmed John solution if
\[
\left|
\left\{
i\in[n]:
x_i^\top M(\mu)x_i\le e^\gamma
\right\}
\right|
\ge (1-\kappa)n
\]
\end{definition}

In particular, a scaling $E(\mu)$ by $e^{-\gamma}$
is contained inside the polytope induced by at least
$(1-\kappa)n$ constraints, while its volume
approximates that of $E^*$.
The corresponding geometric guarantees are
established in
Appendix~\ref{apx:additional-background}.

\paragraph{Dense measures and KL projection.}
For $\kappa\in(0,1)$, define the set of
$\kappa$-dense measures
\[
\Gamma_\kappa:=
\left\{
\mu\in\mathbb{R}_{\ge0}^n:
\sum_{i=1}^n\mu_i=d,\quad
\mu_i\le\frac{d}{\kappa n}
\ \forall i\in[n]
\right\}
\]
For two measures $\mu,\nu$,
their KL divergence is $D_{\mathrm{KL}}(\mu\Vert\nu)
:=
\sum_{i=1}^n
\left[
\mu_i\log\left(\frac{\mu_i}{\nu_i}\right)
+\nu_i-\mu_i
\right]$.
The KL projection onto $\Gamma_\kappa$ is $\Pi_{\Gamma_\kappa}(\nu)
:=
\arg\min_{\mu\in\Gamma_\kappa}
D_{\mathrm{KL}}(\mu\Vert\nu)$.

By the Pythagorean property of Bregman
projections~\citep{bregman1967relaxation},
for every $q\in\Gamma_\kappa$ and measure $\nu$ that sums to $d$
\[
D_{\mathrm{KL}}
(q\Vert\Pi_{\Gamma_\kappa}(\nu))
\le
D_{\mathrm{KL}}(q\Vert\nu)
\]
The explicit projection and its computation
are given in Appendix~\ref{apx:kl-projection}.

\paragraph{Leverage scores.}
For every $i\in[n]$, define $\psi_i(\mu):=x_i^\top M(\mu)x_i,~
\phi_i(\mu):=\log\psi_i(\mu)$.
The function $\phi_i$ is convex in $\mu$
whenever $\Sigma(\mu)$ is positive definite
~\citep{pmlr-v99-cohen19a}.
Moreover, the weighted leverage scores
$\ell_i(\mu):=\mu_i\psi_i(\mu)$ satisfy $0\le\ell_i(\mu)\le1$ and $
\sum_{i=1}^n\ell_i(\mu)=d$.

\paragraph{Differential privacy (DP).}
Two datasets are neighboring if they differ
in at most one constraint.
We use zero-concentrated differential privacy
(zCDP)~\citep{bun2016concentrated}.
A randomized mechanism $\mathcal M$ is
$\rho$-zCDP if, for all neighboring datasets
$P,P'$ and all $\alpha>1$,
\[
D_\alpha(\mathcal M(P)\Vert\mathcal M(P'))
\le \alpha\rho
\]
where $D_\alpha$ denotes the R\'enyi divergence
of order $\alpha$.

We use the Gaussian mechanism and the adaptive
composition property of zCDP. In particular,
the composition of mechanisms with privacy
parameters $\rho_1,\ldots,\rho_T$ satisfies
$(\sum_{t=1}^T\rho_t)$-zCDP.
Furthermore, $\rho$-zCDP implies
$(\varepsilon,\delta)$-DP for $\varepsilon
=
\rho+2\sqrt{\rho\log(1/\delta)}$.

Additional properties of DP and Gaussian concentration
facts are recalled in
Appendix~\ref{apx:additional-background}.

\section{A Non-Private Bi-Criteria Approximating Algorithm for JE}
\label{sec:non-private-alg}

In this section, we present a deterministic algorithm
for approximating the John ellipsoid while allowing
a $\kappa$-fraction of the constraints to be violated.
Our algorithm extends the multiplicative-weights
approach of~\cite{pmlr-v99-cohen19a} by projecting
the intermediate weights onto the set of
$\kappa$-dense measures.
This yields the following guarantee.

\begin{theorem}
\label{thm:main-non-private}
Let $P=\{x_1,\ldots,x_n\}\subseteq\mathbb R^d$
span $\mathbb R^d$.
For any $\gamma,\kappa\in(0,1)$,
Algorithm~\ref{alg:non-private} outputs, after $T=O\left(\frac{\log(1/\kappa)}{\gamma}\right)$
iterations, a positive definite matrix $M$
defining the ellipsoid
\[
E:=\{x\in\mathbb R^d:x^\top M^{-1}x\le1\},
\]
such that (i) there exists $S\subseteq[n]$ with
$|S|\ge(1-\kappa)n$ for which
\[
e^{-\gamma/2}E\subseteq K_S,
\qquad \text{ and }\qquad
K\subseteq\sqrt d\,E,
\]
where
$K_S:=\{x:|x_i^\top x|\le1,\ \forall i\in S\}$; and moreover (ii)
$
\operatorname{vol}(e^{-\gamma/2}E)
\ge
e^{-d\gamma/2}\operatorname{vol}(E^*),
$.
where $E^*$ is the John ellipsoid of $K_P$.
\end{theorem}

\paragraph{The algorithm.}
We maintain an unprojected positive measure $w^t$ together with its
$\kappa$-dense KL projection $\mu^t:=\Pi_{\Gamma_\kappa}(w^t)\in\Gamma_\kappa$.
At each iteration, the projected measure $\mu^t$ determines the
weighted covariance matrix $\Sigma(\mu^t)
=
\sum_{i=1}^n\mu_i^t x_ix_i^\top$
and its inverse $M^t=M(\mu^t)$.
The unprojected weights are then updated multiplicatively according
to the corresponding quadratic scores, $w_i^{t+1}
=
w_i^t x_i^\top M^t x_i.$
Thus, the projection is used to determine the measure supplied to the
covariance oracle, while the multiplicative update is applied to the
unprojected weights.


\begin{algorithm}[htb]
\caption{Non-Private John Ellipsoid Approximation}
\label{alg:non-private}
\begin{algorithmic}[1]
\State \textbf{Input:}
$P=\{x_i\}_{i=1}^n$,
$\gamma,\kappa\in(0,1)$.
\State
$T\gets
\left\lceil\log(1/\kappa)/\gamma\right\rceil$
\State
$w_i^0\gets d/n$
for all $i\in[n]$
\For{$t=0,\ldots,T-1$}
    \State
    $\mu^t\gets
    \Pi_{\Gamma_\kappa}(w^t)$
    \State
    $M^t\gets
    \left(
    \sum_{i=1}^n
    \mu_i^t x_ix_i^\top
    \right)^{-1}$
    \For{$i=1,\ldots,n$}
        \State
        $w_i^{t+1}
        \gets
        w_i^t x_i^\top M^t x_i$
    \EndFor
\EndFor
\State
$\overline\mu\gets
\frac1T\sum_{t=0}^{T-1}\mu^t$
\State
\Return
$\left(
\sum_{i=1}^n
\overline\mu_i x_ix_i^\top
\right)^{-1}$
\end{algorithmic}
\end{algorithm}


\paragraph{Analysis.}
The main ingredient of our analysis is a
KL-potential inequality that compares the
algorithm's iterates with any fixed
$\kappa$-dense measure. 
Recall that
$\phi_i(\mu):=\log(x_i^\top M(\mu)x_i)$.

\begin{lemma}
\label{lem:sum-bound}
Let
$\overline\mu=
\frac1T\sum_{t=0}^{T-1}\mu^t$.
Then, for every $q\in\Gamma_\kappa$,
\[
\sum_{i=1}^n q_i\phi_i(\overline\mu)
\le
\frac{D_{\mathrm{KL}}(q\Vert\mu^0)}{T}
\]
\end{lemma}

\begin{proof}
Define $\theta_i^t:=\log\left(\frac{w_i^t}{\mu_i^0}\right),~
F(\theta):=
\max_{\mu\in\Gamma_\kappa}
\left\{
\langle\mu,\theta\rangle
-D_{\mathrm{KL}}(\mu\Vert\mu^0)
\right\}$.
Since $w_i^t=\mu_i^0e^{\theta_i^t}$ and
$\mu^t=\Pi_{\Gamma_\kappa}(w^t)$, the measure $\mu^t$ attains the
maximum in \[F(\theta^t) = \max_{\mu\in\Gamma_\kappa}
\left\{
\langle\mu, \log\left(\frac{w^t}{\mu^0}\right) \rangle
+\sum_i w_i^t-d-\langle \mu, \log\left(\frac{\mu}{\mu^0}\right)\rangle 
\right\} = \max_{\mu\in\Gamma_\kappa}\left\{ \sum_i w_i^t-d-D_{\mathrm{KL}}(\mu || w^t) \right\} \] 
By the Pythagorean property of KL projection (Theorem~\ref{thm:bregman-pythagorean})
for every
$\nu\in\Gamma_\kappa$,
\[
F(\theta^t) - \left(
\langle \nu,\theta^t\rangle
-D_{\mathrm{KL}}(\nu\Vert\mu^0)\right)
= -D_{\mathrm{KL}}(\mu^t||w^t) + D_{\mathrm{KL}}(\nu||w^t) \geq D_{\mathrm{KL}}(\nu||\mu^t)
\]
Now note that the update rule gives $\theta_i^{t+1}
=
\theta_i^t+\phi_i(\mu^t)$. Hence, 
\begin{align*}
F(\theta^{t+1})
&= \max_{\mu\in \Gamma_{\kappa}}\left\{ \langle\mu, \theta^t \rangle  + \sum_{i}\mu_i \phi_i(\mu^t) - D_{\mathrm{KL}}(\mu\Vert\mu^0)   \right\}
\leq \max_{\mu\in \Gamma_{\kappa}} \left\{ F(\theta^t)-D_{\mathrm{KL}}(\mu\Vert \mu^t) + \sum_{i}\mu_i \phi_i(\mu^t) \right\}
\cr&=
F(\theta^t)
+
\max_{\mu\in\Gamma_\kappa}
\left\{
\sum_i \mu_i \left(\log(e^{\phi_i(\mu^t)})
-\log(\frac{\mu_i}{\mu^t_i})\right)
\right\} = F(\theta^t)
+
\max_{\mu\in\Gamma_\kappa}
\left\{
\sum_i \mu_i \left(-\log(\mu_i/\mu_i^t e^{\phi_i(\mu^t)}\right)
\right\}
\cr&
\stackrel{(\ast)}\le F(\theta^t) + 
d \left(\log(d/\sum_i\mu_i^t x_i^\top M(\mu^t)x_i\right)
\stackrel{(\ast\ast)}\le
F(\theta^t) + d\log(\frac{d}{d}) = F(\theta^t)
\end{align*}
where inequality $(\ast)$ follows from the log-sum inequality and the fact that $\sum_i \mu_i = d$, and inequality $(\ast\ast)$ is due to  $\sum_i \mu_i e^{\phi_i(\mu^t)} = \sum_i\mu_i^t x_i^\top M(\mu^t)x_i=d$. As
$\theta^0=0$, we have $F(\theta^T)\le F(\theta^0)=0$.
Fix any $q\in\Gamma_\kappa$ and note that $\theta_i^T = \sum_{t}\phi_i(\mu^t)$ so
\[
\sum_{t=0}^{T-1}\sum_i q_i\phi_i(\mu^t)
=
\langle q,\theta^T\rangle
\le
F(\theta^T) + D_{\mathrm{KL}}(q\Vert\mu^0) \leq D_{\mathrm{KL}}(q\Vert\mu^0)
\]
Dividing by $T$ and applying Jensen's inequality to the convex
functions $\phi_i$ gives
\[
\sum_i q_i\phi_i(\overline\mu)
\le
\frac1T
\sum_{t=0}^{T-1}
\sum_iq_i\phi_i(\mu^t)
\le
\frac{D_{\mathrm{KL}}(q\Vert\mu^0)}{T}\qedhere
\]
\end{proof}


We now give the summary of the proof of Theorem~\ref{thm:main-non-private}.
The uniform initialization and the density
constraint imply
\[
D_{\mathrm{KL}}(q\Vert\mu^0)
\le d\log(1/\kappa),
\qquad
\forall q\in\Gamma_\kappa
\]
Combining this bound with
Lemma~\ref{lem:sum-bound} and our choice of
$T$ gives
\begin{equation}
\label{eq:bounded_dense_measure}
\sup_{q\in\Gamma_\kappa}
\sum_{i=1}^n \frac{q_i} d
\log\!\left(x_i^\top M(\overline\mu)x_i\right)
\le \frac 1 d \cdot \frac{d\log(1/\kappa)} T \leq \gamma
\end{equation}
In particular, 
sort the $n$ constraints $x_{(1)}, x_{(2)},..., x_{(n)}$ in descending order according to $\phi_i(\bar \mu)$, and 
set $q$ as the uniform (dense) measure over the top $\kappa n$ constraints. Plugging-in this $q$ into Equation~(\ref{eq:bounded_dense_measure}) above we obtain $\phi_{(\kappa n)}(\bar\mu) \leq \sum_{i=1}^{\kappa n} \frac{1} {\kappa n} \phi_{(i)}(\bar\mu) \leq \gamma$,
asserting that 
at least $(1-\kappa)n$
constraints satisfy $x_i^\top M(\overline\mu)x_i\le e^\gamma$.
The corresponding ellipsoid therefore
satisfies the trimmed-containment guarantee
after scaling by $e^{-\gamma/2}$.

Finally, the dual formulation of the John
ellipsoid yields the volume guarantee of Theorem~\ref{thm:main-non-private},
while the normalization
$\sum_i\overline\mu_i=d$ implies
$K\subseteq\sqrt d\,E$ as shown in~\cite{pmlr-v99-cohen19a}.
The complete arguments are provided in
Appendix~\ref{app:non-private-proofs}.

\section{A Differentially Private Algorithm for Approximating the John Ellipsoid}
\label{sec:private-alg}

We now extend the projected multiplicative-weights
algorithm of Section~\ref{sec:non-private-alg}
to the private setting. The algorithm has the same
structure as its non-private counterpart, but replaces
each exact weighted covariance matrix with a Gaussian-
perturbed one before inversion. The projection onto
$\Gamma_\kappa$ bounds the weight assigned to every
constraint, while the $(\kappa,\tau)$-goodness
assumption guarantees that the weighted covariance
matrices remain sufficiently well-conditioned for the
perturbation to have only a small effect on the update.

We obtain the following utility guarantee.

\begin{theorem}
\label{thm:main-private}
Let
$P=\{x_1,\ldots,x_n\}\subseteq\mathbb R^d$
be $(\kappa,\tau)$-good and suppose that
$\|x_i\|_2\le R$ for every $i\in[n]$.
Let $\gamma,\kappa,\beta\in(0,1)$ and $\rho>0$, and set $T
=
\left\lceil
\frac{2\log(1/\kappa)}{\gamma}
\right\rceil,~
\eta=\frac{\gamma}{8}$.
Suppose
\[
n
\ge
\frac{
32dR^2
\left(
\sqrt d+
\sqrt{2\log((T+1)/\beta)}
\right)
}{
\kappa\gamma\tau
\sqrt{2\rho/(T+1)}
}
\]
Then, with probability at least $1-\beta$,
Algorithm~\ref{alg:private-alg} does not return
$\perp$ and outputs a positive definite matrix
$\widehat M$, such that by defining the ellipsoid
\[
\widehat E_0
:=
\left\{
x\in\mathbb R^d:
x^\top\widehat M^{-1}x\le1
\right\}
\]
and its scaling $\widehat E
:=
\sqrt{1-\eta}\,e^{-\gamma/2}\widehat E_0$ we get that there exists $S\subseteq[n]$ with
$|S|\ge(1-\kappa)n$ such that
\[
\widehat E\subseteq K_S,
\qquad
K\subseteq
\sqrt{d(1+\eta)}\,\widehat E_0,
\]
and moreover
\[
\operatorname{vol}(\widehat E)
\ge
e^{-d\gamma/2}
\left(
\frac{1-\eta}{1+\eta}
\right)^{d/2}
\operatorname{vol}(E^*),
\]
where $E^*$ is the John ellipsoid of $K$.
\end{theorem}

\paragraph{The algorithm.}
As in the non-private algorithm, we maintain an unprojected measure
$w^t$ and use its $\kappa$-dense KL projection
$\mu^t=\Pi_{\Gamma_\kappa}(w^t)$ to define the covariance matrix.
The difference is that, at each iteration, we perturb the weighted
covariance matrix $\Sigma(\mu^t)
=
\sum_{i=1}^n\mu_i^t x_ix_i^\top$
with symmetric Gaussian noise and use the inverse of the perturbed
matrix in the multiplicative update. A final independent invocation
of the private oracle is used to release the ellipsoid.


\begin{algorithm}[htb]
\caption{Private John Ellipsoid Approximation}
\label{alg:private-alg}
\begin{algorithmic}[1]
\State \textbf{Input:}
$P=\{x_i\}_{i=1}^n$,
$\gamma,\kappa\in(0,1)$,
$\rho>0$, $\tau>0$, $R>0$.
\State
$T\gets
\left\lceil2\log(1/\kappa)/\gamma\right\rceil$
\State
$\rho_0\gets\rho/(T+1)$
\State
$w_i^0\gets d/n$ for all $i\in[n]$
\For{$t=0,\ldots,T-1$}
    \State
    $\mu^t\gets
    \Pi_{\Gamma_\kappa}(w^t)$ \Comment{Project onto the set of $\kappa$-dense measures}
    \State
    $\widehat M^t\gets
    \Call{Private-Oracle}{P,\mu^t,\rho_0,\kappa,R}$
    \If{$\widehat M^t=\perp$}
        \State \Return $\perp$
    \EndIf
    \For{$i=1,\ldots,n$}
        \State
        $w_i^{t+1}
        \gets
        w_i^t x_i^\top\widehat M^t x_i$
    \EndFor
\EndFor
\State
$\overline\mu\gets
\frac1T\sum_{t=0}^{T-1}\mu^t$
\State
$\widehat M\gets
\Call{Private-Oracle}{P,\overline\mu,\rho_0,\kappa,R}$
\State \Return $\widehat M$
\end{algorithmic}
\end{algorithm}


\begin{algorithm}[tb]
\caption{Private Covariance Oracle}
\label{alg:private-oracle}
\begin{algorithmic}[1]
\Procedure{Private-Oracle}{$P,\mu,\rho_0, \kappa, R$} 
    \State
    $\Sigma\gets
    \sum_{i=1}^n\mu_i x_ix_i^\top$
    \State
    $\sigma\gets
    \frac{4dR^2}{\kappa n\sqrt{2\rho_0}}$
    \State Sample a symmetric Gaussian matrix
    $N$ with scale $\sigma$
    \State $\widehat\Sigma\gets\Sigma+N$
    \If{$\lambda_{\min}(\widehat\Sigma)<\tau/2$} \Comment{$\tau$ is a goodness assumption parameter}
        \State \Return $\perp$
    \EndIf
    \State \Return $\widehat\Sigma^{-1}$
\EndProcedure
\end{algorithmic}
\end{algorithm}

\paragraph{Privacy.}
For any $\mu\in\Gamma_\kappa$, changing
one constraint changes the weighted covariance
matrix by Frobenius norm at most $\frac{4dR^2}{\kappa n}$.
Consequently,
Algorithm~\ref{alg:private-oracle}, when invoked
with a $\kappa$-dense measure $\mu$, satisfies
$\rho_0$-zCDP under the noise calibration above.
The proof is given in
Appendix~\ref{app:private-proofs}.

\paragraph{Utility analysis.}
The first ingredient is that goodness prevents
any $\kappa$-dense weighted covariance matrix
from becoming nearly singular.

\begin{lemma}
[$(\kappa,\tau)$-goodness implies uniform non-degeneracy]
\label{lem:good-assumption}
If $P$ is $(\kappa,\tau)$-good, then
\[
\lambda_{\min}(\Sigma(\mu))
\ge\tau
\qquad
\forall\mu\in\Gamma_\kappa
\]
\end{lemma}

Therefore, whenever the Gaussian perturbation
satisfies
$\|N^t\|_2\le\eta\tau$, we have
\[
(1-\eta)\Sigma^t
\preceq
\Sigma^t+N^t
\preceq
(1+\eta)\Sigma^t
\]
and hence
\[
\frac{1}{1+\eta}M^t
\preceq
\widehat M^t
\preceq
\frac{1}{1-\eta}M^t
\]
In particular, each noisy logarithmic score differs
from its exact counterpart by at most $2\eta$.

The KL-potential argument of
Section~\ref{sec:non-private-alg} is stable under
these perturbations.

\begin{lemma}[Perturbed potential inequality]
\label{lem:private-sum-bound}
Suppose that
$\|N^t\|_2\le\eta\tau$
for every $t=0,\ldots,T-1$, where
$\eta\le1/2$.
Then, for every $q\in\Gamma_\kappa$,
\[
\sum_{i=1}^n
q_i
\log\!\left(
x_i^\top M(\overline\mu)x_i
\right)
\le
\frac{
D_{\mathrm{KL}}(q\Vert\mu^0)
}{T}
+
4d\eta
\]
\end{lemma}

Combining Lemma~\ref{lem:private-sum-bound}
with
$D_{\mathrm{KL}}(q\Vert\mu^0)
\le d\log(1/\kappa)$,
our choice
$T\ge2\log(1/\kappa)/\gamma$,
and $\eta=\gamma/8$ gives
\[
\sup_{q\in\Gamma_\kappa}
\sum_{i=1}^n
\frac{q_i} d
\log\!\left(
x_i^\top M(\overline\mu)x_i
\right)
\le \gamma
\]
Exactly as in the non-private analysis, this
implies that at least $(1-\kappa)n$ constraints
satisfy $x_i^\top M(\overline\mu)x_i\le e^\gamma$,
and therefore
$e^{-\gamma/2}E(\overline\mu)\subseteq K_S$
for some $|S|\ge(1-\kappa)n$.

Finally, Gaussian matrix concentration and a union
bound imply that the sample-size assumption of
Theorem~\ref{thm:main-private} guarantees
$\|N^t\|_2\le\eta\tau$ simultaneously for all
$T+1$ oracle calls with probability at least
$1-\beta$. Applying the same matrix-perturbation
bound to the final oracle call transfers the
containment and volume guarantees from
$M(\overline\mu)$ to the released matrix
$\widehat M$.

The complete proofs, including the matrix
perturbation, perturbed KL argument, Gaussian
concentration bound, and final-release analysis,
are given in Appendix~\ref{app:private-proofs}.

\section{Extensions to Minimum-Volume Enclosing Ellipsoids}
\label{sec:extensions}

Our results extend directly to the centered minimum-volume enclosing
ellipsoid (MVEE) problem through the standard polar-duality correspondence
between maximum-volume inscribed and minimum-volume enclosing
ellipsoids~\citep{Schneider_2013,tod16}. In particular, for
$P=\{x_1,\ldots,x_n\}\subseteq\mathbb R^d$, the polytope
$K=\{x:|x_i^\top x|\le1,\ \forall i\in[n]\}$ is the polar of
$\operatorname{conv}\{\pm x_1,\ldots,\pm x_n\}$, and the polar of the
John ellipsoid of $K$ is the centered MVEE of $P$. 
Namely, if $E^* = \{x: ~~ x^\top M^{-1} x \leq 1\}$ is the JE of $K$, then its polar, $(E^*)^\circ = \{x: ~~ x^\top M x \leq 1\}$, is the MVEE of $P$.
This correspondence
also preserves our trimmed formulation: if the John-ellipsoid guarantee
holds for a subset $S\subseteq[n]$ with $|S|\ge(1-\kappa)n$, then the
polar ellipsoid contains the corresponding $(1-\kappa)$-fraction of the
MVEE input points, with the corresponding $\alpha^{-1}$ volume-approximation guarantee for the MVEE, where $\alpha$ denotes the volume-approximation factor for the John ellipsoid.
We give the details of this reduction in
Appendix~\ref{apx:polar-interpretation}. We comment that since taking the polar is
deterministic post-processing of the output of Algorithm~\ref{alg:private-alg}, the reduction incurs no additional
privacy loss.

Even though the above-mentioned result deals with \emph{centered} convex bodies (since clearly the origin is the center of $\operatorname{conv}\{\pm x_1,\ldots,\pm x_n\}$), it easily extends to the general uncentered MVEE problem using
the standard homogeneous-coordinate lifting of~\cite{tod16}. In order to approximate the MVEE of a convex body $\operatorname{conv}\{y_1, y_2, \ldots,y_n\}$, We simply map
each $y_i\in\mathbb R^d$ to $(y_i,1)\in\mathbb R^{d+1}$ and reduce the
uncentered problem to the centered problem in one additional dimension. While this
reduction is standard, we verify in
Lemma~\ref{lem:uncentered-centered-reduction} that it continues to hold
under our trimmed objective: the same subset of at least
$(1-\kappa)n$ points is preserved, together with the approximation
factor. Thus, combining the lifting with the polar reduction above
extends both our non-private and private guarantees to uncentered MVEEs.
We comment that in order to obtain a private guarantee, the \emph{lifted} dataset must satisfy the
$(\kappa,\tau)$-goodness assumption in dimension $d+1$.
The full trimmed lifting argument is given in
Appendix~\ref{apx:uncentered-centered-reduction}.

\section{Conclusion}
\label{sec:conclusion}

We developed a projected multiplicative-weights framework for approximating
the John ellipsoid under a trimmed feasibility criterion. By restricting
the dual weights to $\kappa$-dense measures, our non-private algorithm
produces an ellipsoid satisfying at least a $(1-\kappa)$-fraction of the
constraints after an $e^{-\gamma}$ rescaling, while approximately preserving the
optimal volume in $O(\log(1/\kappa)/\gamma)$ iterations. We also analyzed
a Gaussian-perturbed variant and showed that, under the
$(\kappa,\tau)$-goodness assumption and sufficiently many points, the perturbation required for privacy preservation gives
analogous utility guarantees. Through polar
duality and homogeneous-coordinate lifting, our results extend to
trimmed centered and uncentered minimum-volume enclosing ellipsoids.

Our analysis highlights the role of dense measures and robust
non-degeneracy in controlling the effect of perturbations on geometric
optimization. An important direction for future work is to determine
whether comparable guarantees can be obtained under
weaker assumptions than $(\kappa,\tau)$-goodness, which
is sufficient but not necessary for a successful execution
of the algorithm. Moreover, while our analysis treats
$\kappa$ and $\tau$ as fixed parameters, our population
characterization in Appendix~\ref{apx:population}
shows that, for any fixed $\kappa$, the goodness parameter
$\tau$ can be related to the lower quantiles of the
underlying distribution. Developing a privacy-preserving,
data-dependent procedure for selecting $\kappa$ and
estimating an appropriate $\tau$, while maintaining
the utility guarantees of our algorithm, remains an
interesting direction for future work. More broadly, it would be
interesting to investigate the applicability of projected
multiplicative-weights methods to other robust geometric optimization
problems.

\subsection*{AI Disclosure Statement}

In this work, we used generative AI tools (ChatGPT by OpenAI)
to assist in the formulation and refinement of mathematical
claims, writing of proofs, and the
discussion and refinement of the theoretical framework and
research methodology. Additionally, we used generative AI
tools to assist with drafting and editing portions of the
manuscript, improving readability and clarity, organizing
the presentation of theoretical results, and suggesting titles
and keywords.

All AI-assisted mathematical statements, arguments, and
proofs were reviewed and verified by the authors. The authors
also reviewed and revised AI-assisted text to ensure its
accuracy, consistency with the research contributions, and
appropriate attribution of existing work. The remaining
tasks requiring disclosure were either not applicable to
this work or were performed without the assistance of
generative AI tools.

The authors take full responsibility for the final content
of this paper, including all mathematical claims, proofs,
and text produced with the assistance of generative AI.

\remove{
\subsubsection*{Acknowledgments}
Use unnumbered third level headings for the acknowledgments. All
acknowledgments, including those to funding agencies, go at the end of the paper.
}
\bibliography{bibliography}
\bibliographystyle{iclr2027_conference}

\appendix

\section{Additional Background}
\label{apx:additional-background}

In this section, we recall the geometric, optimization,
and privacy tools used in the analysis of our algorithms.

\subsection{John Ellipsoid: Primal and Dual Formulations}
\label{apx:john-primal-dual}

\paragraph{Primal formulation.}
Let $P=\{x_1,\ldots,x_n\}\subseteq\mathbb R^d$
satisfy $\operatorname{span}(P)=\mathbb R^d$.
Recall that we consider the centrally symmetric polytope
\[
K_P:=
\left\{
x\in\mathbb R^d:
|x_i^\top x|\le1,\ \forall i\in[n]
\right\}
\]
By symmetry, the John ellipsoid of $K_P$ is
centered at the origin and can be represented as
\[
E(G):=
\left\{
x\in\mathbb R^d:
x^\top G^{-2}x\le1
\right\}
\]
where $G\succ0$.
Its volume is
\[
\operatorname{vol}(E(G))
=
\operatorname{vol}(B_2^d)\det(G)
\]
where $B_2^d$ denotes the Euclidean unit ball.

The containment condition $E(G)\subseteq K_P$
is equivalent to
\[
\max_{x\in E(G)}|x_i^\top x|\le1
\qquad\forall i\in[n]
\]
Writing $x=Gy$ with $\|y\|_2\le1$, we obtain
\[
\max_{x\in E(G)}|x_i^\top x|
=
\max_{\|y\|_2\le1}|x_i^\top Gy|
=
\|Gx_i\|_2
\]
Consequently, the John ellipsoid is obtained
by solving
\[
\begin{aligned}
\text{maximize}\qquad
&\log\det(G^2)\\
\text{subject to}\qquad
&G\succ0,\\
&\|Gx_i\|_2^2\le1,
\qquad\forall i\in[n]
\end{aligned}
\tag{P}
\]

Equivalently, defining $B:=G^2$, we obtain
the convex optimization problem
\[
\begin{aligned}
\text{maximize}\qquad
&\log\det B\\
\text{subject to}\qquad
&B\succ0,\\
&x_i^\top Bx_i\le1,
\qquad\forall i\in[n]
\end{aligned}
\tag{$P'$}
\]
The corresponding ellipsoid is
\[
E(B):=
\{x\in\mathbb R^d:x^\top B^{-1}x\le1\}
\]

\paragraph{Dual formulation.}
Introduce nonnegative Lagrange multipliers
$\mu_1,\ldots,\mu_n$ for the constraints
$x_i^\top Bx_i\le1$.
The Lagrangian of $(P')$ is
\[
\begin{aligned}
L(B,\mu)
&=
\log\det B+
\sum_{i=1}^n
\mu_i(1-x_i^\top Bx_i)\\
&=
\log\det B+\sum_{i=1}^n\mu_i
-\operatorname{tr}(B\Sigma(\mu))
\end{aligned}
\]
where
\[
\Sigma(\mu):=
\sum_{i=1}^n\mu_i x_ix_i^\top.
\]

For $\Sigma(\mu)\succ0$, maximizing the
Lagrangian over $B\succ0$ gives the
stationarity condition $B^{-1}-\Sigma(\mu)=0$
and therefore $B=\Sigma(\mu)^{-1}$.
Substituting into the Lagrangian yields
\[
\begin{aligned}
\sup_{B\succ0}L(B,\mu)
&=
-\log\det\Sigma(\mu)
+\sum_{i=1}^n\mu_i-d
\end{aligned}
\]

Thus, the dual problem is
\[
\begin{aligned}
\text{minimize}\qquad
&D(\mu):=
\sum_{i=1}^n\mu_i
-\log\det\Sigma(\mu)-d\\
\text{subject to}\qquad
&\mu_i\ge0,\qquad\forall i\in[n]\\
&\Sigma(\mu)\succ0
\end{aligned}
\tag{D}
\]
We adopt the convention that
$D(\mu)=+\infty$ when $\Sigma(\mu)$
is singular.

Since the primal problem satisfies
Slater's condition, strong duality holds:
\[
\log\det B^*=D(\mu^*)
\]
where $B^*$ and $\mu^*$ denote optimal
primal and dual solutions, respectively.

\paragraph{John decomposition.}
The optimality conditions imply $B^*=\Sigma(\mu^*)^{-1}$,
and consequently $(B^*)^{-1}
=
\sum_{i=1}^n\mu_i^*x_ix_i^\top$.
This representation is commonly referred
to as the John decomposition.

Furthermore, complementary slackness gives
\[
\mu_i^*
\left(x_i^\top B^*x_i-1\right)=0
\qquad\forall i\in[n]
\]
Taking the trace of
$B^*\Sigma(\mu^*)=I_d$, we obtain
\[
\begin{aligned}
d
&=
\operatorname{tr}(B^*\Sigma(\mu^*))\\
&=
\sum_{i=1}^n
\mu_i^*x_i^\top B^*x_i\\
&=
\sum_{i=1}^n\mu_i^*
\end{aligned}
\]
Hence the optimal dual measure has total
mass $d$, and its objective simplifies to
\[
D(\mu^*)
=
-\log\det\Sigma(\mu^*)
\]

\paragraph{Geometric interpretation of trimmed solutions.}
Recall that a nonnegative measure $\mu$
of total mass $d$ is a
$(1+\gamma)$-approximate $\kappa$-trimmed
John solution if
\[
\left|
\left\{
i\in[n]:
x_i^\top M(\mu)x_i\le e^\gamma
\right\}
\right|
\ge(1-\kappa)n
\]
where $M(\mu):=\Sigma(\mu)^{-1}$.

The following lemma relates this condition
to the geometric guarantees of the
corresponding ellipsoid.

\begin{lemma}
[$(1+\gamma)$-approximate trimmed solution
is a good trimmed rounding]
\label{lem:trimmed-rounding}
Let $\mu$ be a $(1+\gamma)$-approximate
$\kappa$-trimmed John solution, and define
\[
E:=
\left\{
x\in\mathbb R^d:
x^\top M(\mu)^{-1}x\le1
\right\}
\]
Let
\[
S:=
\left\{
i\in[n]:
x_i^\top M(\mu)x_i\le e^\gamma
\right\}
\]
and
\[
K_S:=
\left\{
x\in\mathbb R^d:
|x_i^\top x|\le1,\ \forall i\in S
\right\}
\]
Then $|S|\ge(1-\kappa)n$ and
\[
e^{-\gamma/2}E\subseteq K_S
\]
Moreover, for the original polytope,
\[
K_P\subseteq\sqrt d\,E
\]
\end{lemma}

\begin{proof}
The bound $|S|\ge(1-\kappa)n$ follows
directly from the definition of an
approximate trimmed John solution.

To establish trimmed containment, let
$x\in e^{-\gamma/2}E$.
By definition, $x^\top M(\mu)^{-1}x\le e^{-\gamma}$.
For every $i\in S$, we also have $x_i^\top M(\mu)x_i\le e^\gamma$.
Applying Cauchy-Schwarz gives
\[
\begin{aligned}
|x_i^\top x|
&=
\left|
\left\langle
M(\mu)^{1/2}x_i,
M(\mu)^{-1/2}x
\right\rangle
\right|\\
&\le
\sqrt{x_i^\top M(\mu)x_i}
\sqrt{x^\top M(\mu)^{-1}x}\\
&\le
e^{\gamma/2}e^{-\gamma/2}
=1
\end{aligned}
\]
Therefore, $e^{-\gamma/2}E\subseteq K_S$.

Next, let $x\in K_P$. Then
$|x_i^\top x|\le1$ for every $i\in[n]$.
Since $M(\mu)^{-1}=\Sigma(\mu)$,
\[
\begin{aligned}
x^\top M(\mu)^{-1}x
&=
x^\top
\left(
\sum_{i=1}^n\mu_i x_ix_i^\top
\right)x
=
\sum_{i=1}^n
\mu_i|x_i^\top x|^2
\le
\sum_{i=1}^n\mu_i=d
\end{aligned}
\]
Thus $x\in\sqrt d\,E$, proving $K_P\subseteq\sqrt d\,E$.
\end{proof}

\paragraph{Volume interpretation via dual feasibility.}
The dual formulation provides a lower
bound on the volume of the ellipsoid
associated with any normalized dual
measure.

\begin{lemma}[Volume interpretation via dual feasibility]
\label{lem:volume-interpretation}
Let $\mu\in\mathbb R_{\ge0}^n$ satisfy
\[
\sum_{i=1}^n\mu_i=d,
\qquad
\Sigma(\mu)\succ0
\]
and define
\[
E:=
\left\{
x\in\mathbb R^d:
x^\top M(\mu)^{-1}x\le1
\right\}
\]
Let $E^*$ be the John ellipsoid of $K_P$.
For every $\gamma\ge0$, the scaled
ellipsoid $E_\gamma:=e^{-\gamma/2}E$
satisfies
\[
\operatorname{vol}(E_\gamma)
\ge
e^{-d\gamma/2}\operatorname{vol}(E^*)
\]
\end{lemma}

\begin{proof}
Let $\mu^*$ denote an optimal dual solution.
By weak duality, $D(\mu)\ge D(\mu^*)$.
Since both measures have total mass $d$,
the dual objectives simplify to
\[
D(\mu)=-\log\det\Sigma(\mu)
\]
and
\[
D(\mu^*)=-\log\det\Sigma(\mu^*)
\]
Consequently,
\[
\det\Sigma(\mu)
\le
\det\Sigma(\mu^*)
\]

The volume of $E$ is
\[
\operatorname{vol}(E)
=
\operatorname{vol}(B_2^d)
\det(M(\mu))^{1/2}
=
\frac{\operatorname{vol}(B_2^d)}
{\sqrt{\det\Sigma(\mu)}}
\]
Similarly, by strong duality and the
John decomposition,
\[
\operatorname{vol}(E^*)
=
\frac{\operatorname{vol}(B_2^d)}
{\sqrt{\det\Sigma(\mu^*)}}.
\]
Hence $\operatorname{vol}(E)
\ge\operatorname{vol}(E^*)$.

Finally, scaling a $d$-dimensional
ellipsoid by $e^{-\gamma/2}$ multiplies
its volume by $e^{-d\gamma/2}$.
Therefore,
\[
\begin{aligned}
\operatorname{vol}(E_\gamma)
&=
e^{-d\gamma/2}
\operatorname{vol}(E)\\
&\ge
e^{-d\gamma/2}
\operatorname{vol}(E^*)
\end{aligned}
\]
\end{proof}

The preceding volume comparison does not
by itself imply that $E_\gamma$ is
contained in the original polytope.
The corresponding trimmed-containment
guarantee follows separately from
Lemma~\ref{lem:trimmed-rounding}.

\subsection{Polar Interpretation: MVEE and MVIE}
\label{apx:polar-interpretation}
\label{subsubsec:polar-interpretation}

We recall the relationship between polar bodies,
minimum-volume enclosing ellipsoids (MVEE),
and maximum-volume inscribed ellipsoids (MVIE).
For the centrally symmetric polytopes considered
in this work, the MVIE coincides with the
John ellipsoid.

\begin{definition}[Polar body~\citep{Schneider_2013}]
Let $K\subseteq\mathbb R^d$ be a convex body
containing the origin.
Its polar body is
\[
K^\circ:=
\left\{
y\in\mathbb R^d:
\langle x,y\rangle\le1
\text{ for every }x\in K
\right\}
\]
When $K$ is centrally symmetric, this is
equivalently
\[
K^\circ=
\left\{
y\in\mathbb R^d:
|\langle x,y\rangle|\le1
\text{ for every }x\in K
\right\}
\]
\end{definition}

\begin{fact}[Basic properties of polarity~\citep{Schneider_2013}]
\label{fact:polar-properties}
Let $K_1,K_2$ be convex bodies containing
the origin. Then:
\begin{enumerate}
\item \emph{Containment reversal:}
\[
K_1\subseteq K_2
\quad\Longrightarrow\quad
K_2^\circ\subseteq K_1^\circ
\]

\item \emph{Bipolar identity:}
If $K$ is closed and convex and contains
the origin in its interior, then
\[
(K^\circ)^\circ=K
\]

\item \emph{Scaling:}
For every $\alpha>0$,
\[
(\alpha K)^\circ=\alpha^{-1}K^\circ
\]
\end{enumerate}
\end{fact}

\begin{fact}[Polarity of centered ellipsoids~\citep{tod16}]
\label{fact:polar-ellipsoid}
Let $M\succ0$ and define
\[
E(M):=
\left\{
x\in\mathbb R^d:
x^\top M^{-1}x\le1
\right\}
\]
Then
\[
E(M)^\circ
=
\left\{
y\in\mathbb R^d:
y^\top My\le1
\right\}
\]
In particular, the polar of an
origin-centered ellipsoid is again
an origin-centered ellipsoid.
\end{fact}


\paragraph{John ellipsoid and centered MVEE.}
Given $P=\{x_1,\ldots,x_n\}$, define
the centrally symmetric convex hull
\[
C_P:=
\operatorname{conv}
\{\pm x_1,\ldots,\pm x_n\}
\]
By the definition of polarity,
\[
K_P=C_P^\circ,
\qquad
K_P^\circ=C_P
\]

Let $E^*$ be the John ellipsoid of $K_P$.
Since $E^*\subseteq K_P$, containment
reversal gives
\[
C_P=K_P^\circ\subseteq(E^*)^\circ
\]
Thus $(E^*)^\circ$ is an origin-centered
ellipsoid enclosing all the points of $P$.

Conversely, let $C$ be any
origin-centered ellipsoid satisfying
$C_P\subseteq C$.
Taking polars gives
\[
C^\circ\subseteq C_P^\circ=K_P
\]
Therefore, $C^\circ$ is an admissible
centered inscribed ellipsoid of $K_P$.

For an origin-centered ellipsoid
$E(M)$, we have
\[
\operatorname{vol}(E(M))
=
\operatorname{vol}(B_2^d)
\det(M)^{1/2}
\]
and
\[
\operatorname{vol}(E(M)^\circ)
=
\operatorname{vol}(B_2^d)
\det(M)^{-1/2}
\]
Consequently,
\[
\operatorname{vol}(E)
\operatorname{vol}(E^\circ)
=
\operatorname{vol}(B_2^d)^2
\]

Maximizing the volume of an
origin-centered ellipsoid inscribed
in $K_P$ is therefore equivalent to
minimizing the volume of an
origin-centered ellipsoid enclosing
$C_P$.

It follows that $C^*:=(E^*)^\circ$
is the centered minimum-volume enclosing
ellipsoid of $C_P$, and equivalently
the centered MVEE of $P$.

\paragraph{Trimmed containment under polarity.}
Let $S\subseteq[n]$ and define
\[
K_S:=
\left\{
x:
|x_i^\top x|\le1
\quad\forall i\in S
\right\}
\]
Suppose that an origin-centered ellipsoid
$\widehat E$ satisfies $\widehat E\subseteq K_S$.
Taking polars yields
\[
\operatorname{conv}
\{\pm x_i:i\in S\}
=
K_S^\circ
\subseteq\widehat E^\circ
\]
In particular,
\[
x_i\in\widehat E^\circ
\qquad\forall i\in S
\]
Thus, if $|S|\ge(1-\kappa)n$,
the polar ellipsoid encloses at least
a $(1-\kappa)$-fraction of the input points.

Moreover, if $\operatorname{vol}(\widehat E)
\ge
\alpha\operatorname{vol}(E^*)$,
then
\[
\begin{aligned}
\frac{
\operatorname{vol}(\widehat E^\circ)
}{
\operatorname{vol}(C^*)
}
&=
\frac{
\operatorname{vol}(E^*)
}{
\operatorname{vol}(\widehat E)
}\le\alpha^{-1}
\end{aligned}
\]
Hence, the trimmed-containment and
volume guarantees for the John ellipsoid
translate directly into corresponding
guarantees for the centered MVEE problem.

The extension to general uncentered MVEEs
through homogeneous-coordinate lifting
is discussed separately in
Section~\ref{apx:uncentered-centered-reduction}.

\subsection{Additional Differential Privacy
and Concentration Facts}
\label{apx:privacy-background}

We recall the standard privacy definitions,
composition properties, and Gaussian
concentration inequalities used in the
analysis of our private algorithm.

\paragraph{Differential privacy.}

\begin{definition}
[Differential Privacy~\citep{dwork2006calibrating}]
A randomized mechanism
$\mathcal M:\mathcal X^n\to\mathcal Y$
is $(\varepsilon,\delta)$-differentially
private if, for every pair of neighboring
datasets $P,P'\in\mathcal X^n$ and every
measurable set $S\subseteq\mathcal Y$,
\[
\Pr[\mathcal M(P)\in S]
\le
e^\varepsilon
\Pr[\mathcal M(P')\in S]+\delta
\]
When $\delta=0$, the mechanism is called
$\varepsilon$-differentially private.
\end{definition}

\begin{definition}
[Zero-Concentrated Differential Privacy
(zCDP)~\citep{bun2016concentrated}]
A randomized mechanism $\mathcal M$
satisfies $\rho$-zCDP if, for every pair
of neighboring datasets $P,P'$ and every
$\alpha>1$,
\[
D_\alpha
\bigl(
\mathcal M(P)\Vert\mathcal M(P')
\bigr)
\le\alpha\rho
\]
where $D_\alpha$ is the Renyi divergence
of order $\alpha$.
\end{definition}

The following standard properties of zCDP
are used throughout the paper.

\begin{fact}[Composition and post-processing]
\label{fact:zcdp-composition}
Suppose that $\mathcal M_t$ satisfies
$\rho_t$-zCDP for
$t=1,\ldots,T$.
If each mechanism satisfies its claimed
privacy guarantee conditionally on the
previously released outputs, then their
adaptive composition satisfies $\left(\sum_{t=1}^T\rho_t\right)
\text{-zCDP}$.

Moreover, arbitrary randomized
post-processing of the output of a
$\rho$-zCDP mechanism preserves
the same privacy guarantee.
\end{fact}

\begin{fact}[Conversion to approximate DP]
\label{fact:zcdp-conversion}
A $\rho$-zCDP mechanism also satisfies
$(\varepsilon,\delta)$-DP for every
$\delta\in(0,1)$, where $\varepsilon
=
\rho+
2\sqrt{\rho\log(1/\delta)}$.
\end{fact}

\paragraph{Gaussian mechanism.}
Let $f:\mathcal X^n\to\mathbb R^m$
have global $\ell_2$-sensitivity
\[
\Delta_2(f):=
\sup_{P\sim P'}
\|f(P)-f(P')\|_2
\]

\begin{fact}[Gaussian mechanism~\citep{bun2016concentrated}]
\label{fact:gaussian-mechanism}
Let $Z\sim\mathcal N(0,\sigma^2I_m)$.
The mechanism
\[
\mathcal M(P):=f(P)+Z
\]
satisfies $\rho$-zCDP whenever
\[
\sigma^2
\ge
\frac{\Delta_2(f)^2}{2\rho}
\]
\end{fact}

\paragraph{Gaussian perturbation of symmetric matrices.}
We identify the vector space of real
symmetric $d\times d$ matrices with
$\mathbb R^{d(d+1)/2}$ through an
orthonormal basis under the Frobenius
inner product
\[
\langle A,B\rangle_F
:=
\operatorname{tr}(A^\top B)
\]

Let $F_{ii}:=e_ie_i^\top$,
and, for $i<j$, let $F_{ij}:=
\frac{e_ie_j^\top+e_je_i^\top}{\sqrt2}$.
The collection
$\{F_{ii}\}_{i=1}^d
\cup\{F_{ij}\}_{i<j}$
is an orthonormal basis of the real
symmetric matrices.

A symmetric Gaussian perturbation with
scale $\sigma$ can be written as
\[
N=
\sum_{i=1}^d z_{ii}F_{ii}
+
\sum_{1\le i<j\le d}z_{ij}F_{ij}
\]
where the coefficients $z_{ij}$ are
independent $\mathcal N(0,\sigma^2)$
random variables.

Equivalently, its diagonal entries have
variance $\sigma^2$, and its
off-diagonal entries have variance
$\sigma^2/2$.

Under this normalization, perturbing a
symmetric matrix-valued query with $N$
is equivalent to applying the standard
Gaussian mechanism to its coordinates
in an orthonormal basis.

In particular, a symmetric matrix-valued
query with global Frobenius sensitivity
$\Delta_F$ satisfies $\rho$-zCDP
when perturbed using $\sigma^2
\ge
\frac{\Delta_F^2}{2\rho}$.

\paragraph{Sensitivity of a weighted covariance matrix.}
Consider the weighted covariance query
\[
\Sigma(P,\mu)
:=
\sum_{i=1}^n\mu_i x_ix_i^\top
\]
where $\mu\in\Gamma_\kappa$ is fixed
and $\|x_i\|_2\le R$ for every $i$.

Let $P$ and $P'$ differ only in their
$j$-th constraint, where $x_j$ is
replaced by $x_j'$.
Then
\[
\begin{aligned}
&\|\Sigma(P,\mu)-\Sigma(P',\mu)\|_F\\
&\qquad=
\mu_j
\|x_jx_j^\top-x_j'x_j'^\top\|_F\\
&\qquad\le
\mu_j
\left(
\|x_jx_j^\top\|_F+
\|x_j'x_j'^\top\|_F
\right)\\
&\qquad\le
2\mu_jR^2
\le
\frac{2dR^2}{\kappa n}
\end{aligned}
\]

Thus, for a fixed measure $\mu$, the
weighted covariance query has global
Frobenius sensitivity at most $\Delta_F
\le
\frac{2dR^2}{\kappa n}$.

\paragraph{Gaussian norm concentration.}

\begin{fact}[Gaussian vector concentration]
\label{fact:gaussian-concentration}
Let $Z\sim\mathcal N(0,\sigma^2I_d)$.
Then, for every $u>0$,
\[
\Pr\left[
\|Z\|_2
\ge
\sigma\left(\sqrt d+\sqrt{2u}\right)
\right]
\le e^{-u}
\]
Equivalently, for every $\beta\in(0,1)$,
\[
\Pr\left[
\|Z\|_2
\le
\sigma
\left(
\sqrt d+\sqrt{2\log(1/\beta)}
\right)
\right]
\ge1-\beta
\]
\end{fact}

This bound follows from the concentration
of the Euclidean norm of a standard
Gaussian vector, whose squared norm
follows a $\chi_d^2$ distribution.

For the symmetric Gaussian matrix
described above, a corresponding
spectral-norm concentration bound takes
the form
\[
\Pr\left[
\|N\|_2>
C\sigma(\sqrt d+\sqrt u)
\right]
\le e^{-u}
\]
where $C>0$ is an absolute constant.

The precise constant depends on the
normalization of the symmetric Gaussian
matrix distribution.

In particular, for independent symmetric
Gaussian perturbations
$N^0,\ldots,N^T$, a union bound gives
\[
\Pr\left[
\max_{0\le t\le T}\|N^t\|_2
\le
C\sigma
\left(
\sqrt d+
\sqrt{\log\frac{T+1}{\beta}}
\right)
\right]
\ge1-\beta
\]

This concentration inequality is used
in Appendix~\ref{app:private-proofs}
to obtain the sample-size condition
ensuring that the noisy covariance
matrices remain well-conditioned.

\subsection{Additional Facts on KL Projection
and Leverage Scores}
\label{apx:kl-leverage-background}

\paragraph{Measures and generalized KL divergence.}
A nonnegative measure over the finite
domain $[n]$ is a vector
$\mu\in\mathbb R_{\ge0}^n$.
Throughout the paper, the iterates
belong to the set of measures of total
mass $d$.

For nonnegative measures $\mu,\nu$,
with $\nu_i>0$ for all $i$, we define
the generalized KL divergence by
\[
D_{\mathrm{KL}}(\mu\Vert\nu)
:=
\sum_{i=1}^n
\left[
\mu_i\log\left(\frac{\mu_i}{\nu_i}\right)
+\nu_i-\mu_i
\right]
\]
with the convention $0\log0=0$.

When both measures have total mass $d$,
the linear terms cancel, and
\[
D_{\mathrm{KL}}(\mu\Vert\nu)
=
\sum_{i=1}^n
\mu_i\log\left(\frac{\mu_i}{\nu_i}\right)
\]

Recall that the set of
$\kappa$-dense measures is
\[
\Gamma_\kappa:=
\left\{
\mu\in\mathbb R_{\ge0}^n:
\sum_{i=1}^n\mu_i=d,\quad
\mu_i\le\frac{d}{\kappa n}
\ \forall i\in[n]
\right\}
\]

\paragraph{Bregman projection.}
Let $\Gamma$ be a nonempty closed convex
set of nonnegative measures with total
mass $d$.
The KL projection of a strictly
positive measure $\widehat\mu$ onto
$\Gamma$ is
\[
\Pi_\Gamma(\widehat\mu)
:=
\arg\min_{\mu\in\Gamma}
D_{\mathrm{KL}}(\mu\Vert\widehat\mu)
\]

The following standard result gives
the property of KL projection used
in the analysis of the multiplicative
updates.

\begin{theorem}
[Bregman Pythagorean inequality
~\citep{bregman1967relaxation}]
\label{thm:bregman-pythagorean}
Let $\Gamma$ be a nonempty closed convex
set of measures of total mass $d$,
and let $\widehat\mu$ be a strictly
positive measure.
For every $q\in\Gamma$,
\[
\begin{aligned}
&D_{\mathrm{KL}}
(q\Vert\Pi_\Gamma(\widehat\mu))+
D_{\mathrm{KL}}
(\Pi_\Gamma(\widehat\mu)\Vert\widehat\mu)\le
D_{\mathrm{KL}}(q\Vert\widehat\mu)
\end{aligned}
\]
In particular,
\[
D_{\mathrm{KL}}
(q\Vert\Pi_\Gamma(\widehat\mu))
\le
D_{\mathrm{KL}}(q\Vert\widehat\mu)
\]
\end{theorem}
Note that the inequality is (somewhat surprisingly) in the opposite direction to the standard triangle inequality.
\begin{proof}
Let
$p:=\Pi_\Gamma(\widehat\mu)$.
The first-order optimality condition
for the convex minimization defining
$p$ implies that, for every $q\in\Gamma$,
\[
\left\langle
q-p,
\nabla F(p)-\nabla F(\widehat\mu)
\right\rangle
\ge0
\]
where
\[
F(\mu):=
\sum_{i=1}^n(\mu_i\log\mu_i-\mu_i)
\]

The three-point identity for the
Bregman divergence generated by $F$
gives
\[
\begin{aligned}
&D_{\mathrm{KL}}(q\Vert\widehat\mu)
-D_{\mathrm{KL}}(q\Vert p)\\
&\qquad-
D_{\mathrm{KL}}(p\Vert\widehat\mu)\\
&=
\left\langle
q-p,
\nabla F(p)-\nabla F(\widehat\mu)
\right\rangle
\ge0
\end{aligned}
\]
Rearranging proves the first inequality.
The second follows from the
nonnegativity of KL divergence.
\end{proof}

The explicit computation of
$\Pi_{\Gamma_\kappa}$ is established
in the separate KL-projection appendix.
In particular, for strictly positive
$\widehat\mu$, it takes the form
\[
\Pi_{\Gamma_\kappa}(\widehat\mu)_i
=
\min\left\{
\frac{d}{\kappa n},
c\widehat\mu_i
\right\}
\]
where $c>0$ is chosen to ensure that
the projected measure has total mass $d$.

\paragraph{Weighted covariance and leverage scores.}
For a measure $\mu\in\mathbb R_{\ge0}^n$
such that $\Sigma(\mu)\succ0$, recall
\[
\Sigma(\mu):=
\sum_{j=1}^n\mu_jx_jx_j^\top,
\qquad
M(\mu):=\Sigma(\mu)^{-1}
\]
The corresponding quadratic and
logarithmic scores are
\[
\psi_i(\mu):=
x_i^\top M(\mu)x_i,
\qquad
\phi_i(\mu):=\log\psi_i(\mu)
\]

The logarithmic scores are considered
for nonzero constraint vectors, so
that $\psi_i(\mu)>0$.

The following convexity property is
central to our multiplicative-weights
analysis.

\begin{lemma}
[Convexity~\citep{pmlr-v99-cohen19a}]
\label{lem:convexity}
For every nonzero $x_i$, the function
\[
\phi_i(\mu)
=
\log\left(
x_i^\top
\left(
\sum_{j=1}^n\mu_jx_jx_j^\top
\right)^{-1}
x_i
\right)
\]
is convex on the domain where
$\Sigma(\mu)\succ0$.
\end{lemma}







\paragraph{Basic properties of leverage scores.}
Define the weighted leverage score
\[
\ell_i(\mu):=
\mu_i\psi_i(\mu)
=
\mu_i x_i^\top\Sigma(\mu)^{-1}x_i
\]

\begin{lemma}[Basic properties of weighted leverage scores]
\label{lem:leverage-score-properties}
For every $\mu\in\mathbb R_{\ge0}^n$
such that $\Sigma(\mu)\succ0$,
\[
0\le\ell_i(\mu)\le1
\qquad\forall i\in[n]
\]
and
\[
\sum_{i=1}^n\ell_i(\mu)=d
\]
\end{lemma}

\begin{proof}
Nonnegativity follows immediately
from $\mu_i\ge0$ and
$\Sigma(\mu)^{-1}\succ0$.

For the upper bound, observe that
\[
\Sigma(\mu)-\mu_i x_ix_i^\top
=
\sum_{j\ne i}\mu_jx_jx_j^\top
\succeq0
\]
Therefore, $\mu_i x_ix_i^\top
\preceq\Sigma(\mu)$.
Multiplying on both sides by
$\Sigma(\mu)^{-1/2}$ gives
\[
\mu_i
\Sigma(\mu)^{-1/2}
x_ix_i^\top
\Sigma(\mu)^{-1/2}
\preceq I_d
\]
The matrix on the left has at most
one nonzero eigenvalue, equal to
\[
\mu_i x_i^\top\Sigma(\mu)^{-1}x_i
=
\ell_i(\mu)
\]
Hence $\ell_i(\mu)\le1$.

Finally,
\[
\begin{aligned}
\sum_{i=1}^n\ell_i(\mu)
&=
\sum_{i=1}^n
\mu_i x_i^\top\Sigma(\mu)^{-1}x_i
=
\operatorname{tr}\left(
\Sigma(\mu)^{-1}
\sum_{i=1}^n\mu_i x_ix_i^\top
\right)
=
\operatorname{tr}(I_d)=d
\end{aligned}
\]
\end{proof}

In particular, the identity $\sum_{i=1}^n\mu_i\psi_i(\mu)=d$
will be used to control the potential change of the two-level
multiplicative-weights update. Although the unprojected measure
$w^t$ need not have total mass $d$, its projection
$\mu^t=\Pi_{\Gamma_\kappa}(w^t)$ does, and therefore $\sum_{i=1}^n
\mu_i^t\psi_i(\mu^t)=d$
at every iteration.


\section{KL projection onto $\Gamma_\kappa$}
\label{apx:kl-projection}
\begin{lemma}
\label{lem:kl-projection-gamma}
Let $\widehat{\mu}\in\mathbb R_{>0}^n$ be an arbitrary strictly
positive measure, and let
\[
\Gamma_\kappa
=
\left\{
\mu\in\mathbb R_{\ge0}^n:
\sum_{i=1}^n\mu_i=d,
\quad
\mu_i\le \frac{d}{\kappa n}
\ \forall i
\right\}
\]
Define $u:=\frac{d}{\kappa n}$.

Then the Bregman projection of $\widehat{\mu}$ onto $\Gamma_\kappa$
with respect to the generalized KL divergence
\[
\mathrm{KL}(\mu\|\widehat{\mu})
=
\sum_{i=1}^n
\left[
\mu_i\log\left(\frac{\mu_i}{\widehat{\mu}_i}\right)
+
\widehat{\mu}_i-\mu_i
\right]
\]
is uniquely given by
\[
\mu_i^\star
=
\min\left\{
u,\,
c\widehat{\mu}_i
\right\}
\]
where $c>0$ is chosen so that $\sum_{i=1}^n
\mu_i^\star
=
d$.
\end{lemma}

\begin{proof}
We consider the convex optimization problem
\[
\min_{\mu\in\mathbb R_{\ge0}^n}
\mathrm{KL}(\mu\|\widehat{\mu})
\]
subject to
\[
\sum_{i=1}^n\mu_i=d
\qquad \text{and} \qquad
\mu_i\le u
\qquad\forall i\in[n]
\]

Since $\widehat{\mu}_i>0$ for every $i$, the objective is finite and
strictly convex on $\mathbb R_{>0}^n$. Moreover, $\Gamma_\kappa$ is a
nonempty compact convex set. Hence the problem admits a unique minimizer.

Ignoring terms that do not depend on $\mu$, the objective can be written as
\[
f(\mu)
=
\sum_{i=1}^n
\left[
\mu_i\log\left(\frac{\mu_i}{\widehat{\mu}_i}\right)
-\mu_i
\right]
\]

For each coordinate,
\[
\frac{\partial f}{\partial\mu_i}
=
\log\left(\frac{\mu_i}{\widehat{\mu}_i}\right)
\]

Introduce a Lagrange multiplier $\lambda\in\mathbb R$ for the equality
constraint $\sum_i\mu_i=d$, 
and multipliers $\alpha_i\ge0$ for the upper-bound constraints $\mu_i-u\le0$.

Since the optimum has positive coordinates when
$\widehat{\mu}_i>0$, we may omit the nonnegativity multipliers.
The Lagrangian is therefore
\[
L(\mu,\lambda,\alpha)
=
f(\mu)
+
\lambda\left(\sum_{i=1}^n\mu_i-d\right)
+
\sum_{i=1}^n\alpha_i(\mu_i-u)
\]

The KKT stationarity condition gives
\[
\log\left(\frac{\mu_i^\star}{\widehat{\mu}_i}\right)
+
\lambda
+
\alpha_i
=
0
\]
Thus
\[
\mu_i^\star
=
\widehat{\mu}_i
e^{-\lambda-\alpha_i}
\]

We now distinguish two cases.

\begin{itemize}
    \item Case 1: the coordinate is not capped.
    Suppose $\mu_i^\star<u.$
By complementary slackness,
\[
\alpha_i(\mu_i^\star-u)=0
\]
and hence $\alpha_i=0$.
Therefore, $\mu_i^\star
=
e^{-\lambda}\widehat{\mu}_i$.

Now define $c:=e^{-\lambda}>0$. Then every uncapped coordinate satisfies
\[
\mu_i^\star=c\widehat{\mu}_i
\]

    \item Case 2: the coordinate is capped. Suppose $\mu_i^\star=u$.
The stationarity condition gives
\[
\log\left(\frac{u}{\widehat{\mu}_i}\right)
+
\lambda
+
\alpha_i
=
0
\]
Since $\alpha_i\ge0$,
\[
\log\left(\frac{u}{\widehat{\mu}_i}\right)
+\lambda
\le0
\]
Equivalently,
\[
u
\le
e^{-\lambda}\widehat{\mu}_i
=
c\widehat{\mu}_i
\]
Thus a capped coordinate is precisely one for which $c\widehat{\mu}_i\ge u$.

\end{itemize}

Combining the two cases, we obtain
\[
\mu_i^\star
=
\begin{cases}
c\widehat{\mu}_i,
&
c\widehat{\mu}_i<u
\\[3pt]
u,
&
c\widehat{\mu}_i\ge u
\end{cases}
\]
Hence $\mu_i^\star
=
\min\{u,c\widehat{\mu}_i\}$.

The scalar $c$ is determined by the equality constraint
\[
\sum_{i=1}^n\mu_i^\star=d
\]
Therefore it must satisfy
\[
\sum_{i=1}^n
\min\{u,c\widehat{\mu}_i\}
=
d
\]

It remains to show that such a $c$ exists and is unique.

Define
\[
g(c)
:=
\sum_{i=1}^n
\min\{u,c\widehat{\mu}_i\}
\]
Since every $\widehat{\mu}_i>0$, the function $g$ is continuous and
nondecreasing. Moreover, $g(0)=0$, while
\[
\lim_{c\to\infty}g(c)
=
nu
=
\frac{d}{\kappa}
\ge d
\]
Hence, by continuity, there exists $c>0$ such that $g(c)=d$.

If $\kappa<1$, then $nu=\frac{d}{\kappa}>d$, and $g$ is strictly increasing at every point for which $g(c)<nu$.
Therefore the solution of $g(c)=d$ is unique.


Thus the unique KL projection is
\[
\Pi_{\Gamma_\kappa}(\widehat{\mu})_i
=
\min\left\{
\frac{d}{\kappa n},
c\widehat{\mu}_i
\right\}
\]
where $c$ is chosen so that the projected measure has total mass $d$.

\end{proof}

\section{Proofs for the Non-Private Algorithm}
\label{app:non-private-proofs}

We provide the auxiliary results and the complete
proof of Theorem~\ref{thm:main-non-private}.

\subsection{KL-Divergence and Iteration Bounds}

\begin{lemma}
\label{lem:max-kl}
Let $\mu_i^0=d/n$ for all $i\in[n]$.
Then
\[
\max_{q\in\Gamma_\kappa}
D_{\mathrm{KL}}(q\Vert\mu^0)
\le d\log(1/\kappa)
\]
\end{lemma}

\begin{proof}
For every $q\in\Gamma_\kappa$,
the normalization $\sum_i q_i=d$ gives
\[
D_{\mathrm{KL}}(q\Vert\mu^0)
=
\sum_{i=1}^n q_i
\log\left(\frac{q_i}{d/n}\right)
\]
Since $q_i\le d/(\kappa n)$,
\[
D_{\mathrm{KL}}(q\Vert\mu^0)
\le
\sum_{i=1}^n q_i\log(1/\kappa)
=
d\log(1/\kappa)
\]
\end{proof}

\begin{corollary}
\label{cor:iter-bound}
Let $\overline\mu=
\frac1T\sum_{t=0}^{T-1}\mu^t$.
If $T\ge\frac{\log(1/\kappa)}{\gamma}$,
then, for every $q\in\Gamma_\kappa$,
\[
\sum_{i=1}^n \frac{q_i}d\phi_i(\overline\mu)
\le \gamma
\]
\end{corollary}

\begin{proof}
By Lemmas~\ref{lem:sum-bound}
and~\ref{lem:max-kl},
\[
\sum_{i=1}^n \frac{q_i}d\phi_i(\overline\mu)
\le
\frac{D_{\mathrm{KL}}(q\Vert\mu^0)}{dT}
\le
\frac{d\log(1/\kappa)}{dT}
\le \gamma
\]
\end{proof}

\subsection{Trimmed Containment and Volume}

\begin{lemma}[Trimmed containment guarantee]
\label{lem:trimmed-containment}
Let $\mu$ be a measure of total mass $d$
such that
\[
\sum_{i=1}^n
q_i\log\left(x_i^\top M(\mu)x_i\right)
\le d\gamma
\qquad
\forall q\in\Gamma_\kappa
\]
Then the set
\[
S:=
\left\{
i\in[n]:
x_i^\top M(\mu)x_i\le e^\gamma
\right\}
\]
satisfies $|S|\ge(1-\kappa)n$.
\end{lemma}

\begin{proof}
Suppose, towards a contradiction, that more
than $\kappa n$ indices satisfy
$x_i^\top M(\mu)x_i>e^\gamma$.
Let $B$ denote the set of these indices
and define
\[
q_i:=\frac{d}{|B|}\mathbf 1_{\{i\in B\}}
\]
Since $|B|>\kappa n$, we have
$q_i\le d/(\kappa n)$, and hence
$q\in\Gamma_\kappa$.

However, for every $i\in B$,
$\log(x_i^\top M(\mu)x_i)>\gamma$.
Therefore,
\[
\sum_{i=1}^n
q_i\log\left(x_i^\top M(\mu)x_i\right)
=
\frac d{|B|}
\sum_{i\in B}
\log\left(x_i^\top M(\mu)x_i\right)
>d\gamma
\]
contradicting the assumption.
Thus $|S|\ge(1-\kappa)n$.
\end{proof}

\begin{corollary}
[Approximate $\kappa$-trimmed John ellipsoid]
\label{cor:approx-trimmed-john}
Suppose that a measure $\mu$ of total mass $d$
satisfies
\[
\sum_{i=1}^n \frac{q_i}d
\log\left(x_i^\top M(\mu)x_i\right)
\le \gamma
\qquad
\forall q\in\Gamma_\kappa
\]
Let
\[
E(\mu):=
\{x:x^\top M(\mu)^{-1}x\le1\},
\qquad
E_\gamma:=e^{-\gamma/2}E(\mu)
\]
Then there exists $S\subseteq[n]$ with
$|S|\ge(1-\kappa)n$ such that
\[
E_\gamma\subseteq K_S,
\qquad
K_P\subseteq\sqrt d\,E(\mu)
\]
and
\[
\operatorname{vol}(E_\gamma)
\ge
e^{-d\gamma/2}\operatorname{vol}(E^*)
\]
\end{corollary}

\begin{proof}
By Lemma~\ref{lem:trimmed-containment},
there exists $S\subseteq[n]$ with
$|S|\ge(1-\kappa)n$ such that
\[
x_i^\top M(\mu)x_i\le e^\gamma
\qquad\forall i\in S.
\]
For every $x\in E_\gamma$,
\[
x^\top M(\mu)^{-1}x\le e^{-\gamma}
\]
Thus, by Cauchy--Schwarz, for every $i\in S$,
\[
\begin{aligned}
|x_i^\top x|
&=
|\langle M(\mu)^{1/2}x_i,
M(\mu)^{-1/2}x\rangle|
\le
\sqrt{x_i^\top M(\mu)x_i}
\sqrt{x^\top M(\mu)^{-1}x}
\le e^{\gamma/2}e^{-\gamma/2}=1
\end{aligned}
\]
Hence $E_\gamma\subseteq K_S$.

Next, for any $x\in K_P$,
\[
\begin{aligned}
x^\top M(\mu)^{-1}x
&=
\sum_{i=1}^n\mu_i|x_i^\top x|^2
\le\sum_{i=1}^n\mu_i=d
\end{aligned}
\]
Consequently, $K_P\subseteq\sqrt d\,E(\mu)$.

For the volume guarantee, recall that the
John-ellipsoid dual objective is
\[
D(\mu)
=
\sum_{i=1}^n\mu_i
-\log\det\Sigma(\mu)-d
\]
Since $\mu$ has total mass $d$,
\[
D(\mu)=-\log\det\Sigma(\mu)
\]
Let $\mu^*$ denote an optimal dual solution.
By weak duality and strong duality at the
optimum,
\[
-\log\det\Sigma(\mu)
\ge
-\log\det\Sigma(\mu^*)
\]
Since
$\operatorname{vol}(E(\mu))$
is proportional to
$\det(\Sigma(\mu))^{-1/2}$,
we obtain
\[
\operatorname{vol}(E(\mu))
\ge\operatorname{vol}(E^*)
\]
Finally,
\[
\operatorname{vol}(E_\gamma)
=
e^{-d\gamma/2}\operatorname{vol}(E(\mu))
\ge
e^{-d\gamma/2}\operatorname{vol}(E^*)
\]
\end{proof}

\subsection{Proof of the Main Non-Private Guarantee}

\begin{proof}[Proof of
Theorem~\ref{thm:main-non-private}]
Algorithm~\ref{alg:non-private} initializes
$\mu_i^0=d/n$ and performs $T=
\left\lceil
\frac{\log(1/\kappa)}{\gamma}
\right\rceil$
iterations.

By Corollary~\ref{cor:iter-bound},
the averaged measure $\overline\mu$
satisfies
\[
\sum_{i=1}^n
\frac{q_i}d\log\left(
x_i^\top M(\overline\mu)x_i
\right)
\le \gamma
\qquad
\forall q\in\Gamma_\kappa
\]
Moreover, each iterate belongs to
$\Gamma_\kappa$, which is convex.
Hence $\overline\mu\in\Gamma_\kappa$,
and in particular
$\sum_i\overline\mu_i=d$.

Applying Corollary~\ref{cor:approx-trimmed-john}
to $\overline\mu$ gives the claimed
trimmed containment, outer containment,
and volume guarantees for the output
$M=M(\overline\mu)$.
The claimed iteration complexity follows
from the definition of $T$.
\end{proof}

\section{Proofs for the Private Algorithm}
\label{app:private-proofs}

We provide the details underlying the privacy and
utility analysis of Algorithm~\ref{alg:private-alg}.
We first analyze the sensitivity of the covariance
oracle, and then show that sufficiently small
Gaussian perturbations preserve the
multiplicative-weights potential argument.

\subsection{Privacy of the Main Algorithm}

\begin{lemma}[Privacy]
    Algorithm~\ref{alg:private-alg} is $\rho$-zCDP.
\end{lemma}
\begin{proof}
We first show that Algorithm~\ref{alg:private-oracle} is $\rho$-zCDP. Let $P$ and $P'$ be two neighboring datasets of size $n$ that differ in a single point, with $x_i\in P$ replaced by $x_i'\in P'$. Assuming that $\|x_i\|_2,\|x_i'\|_2\le R$, the corresponding weighted covariance matrices can differ in spectral norm by at most
\[
\begin{aligned}
\| \sum\limits_{j} \mu_j(P)x_jx_j^\top - \sum\limits_{j} \mu_j(P')x_jx_j^\top \|_F
&=
\| \sum\limits_{j} (\mu_j(P) - \mu_j(P'))x_jx_j^\top + \mu_i(P')x_ix_i^\top \|_F\\
&\le
R^2\| \mu(P) - \mu(P') \|_1 + \frac{2dR^2}{\kappa n} \leq \frac{4dR^2}{\kappa n}
\end{aligned} 
\]
where the last inequality follows from the bound on the TV-distance between the neighboring $\kappa$-dense measures $\mu$ and $\mu'$, as proven in~\citet{bun2020efficientnoisetolerantprivatelearning} (Lemma 25).
where the last inequality follows from $\mu\in\Gamma_\kappa$, and hence
$\mu_i\le d/(\kappa n)$ for every $i\in[n]$.

Algorithm~\ref{alg:private-alg} invokes Algorithm~\ref{alg:private-oracle} a total of $T+1$ times, with each invocation satisfying $\rho/(T+1)$-zCDP. The final invocation uses the averaged measure $\overline{\mu}
=
\frac1T\sum_{t=1}^T\mu^{t}$.
Since each $\mu^{t}\in\Gamma_\kappa$ and $\Gamma_\kappa$ is convex, we also have
$\overline{\mu}\in\Gamma_\kappa$, so the same sensitivity bound applies to the final release. The claimed $\rho$-zCDP guarantee therefore follows from sequential composition of zCDP.
\end{proof}

\subsection{Goodness and Matrix Perturbation}

\begin{proof}[Proof of Lemma~\ref{lem:good-assumption}]
Fix $\mu\in\Gamma_\kappa$ and a unit vector
$v\in\mathbb S^{d-1}$.
Let
\[
p_i:=\frac{\mu_i}{d},
\qquad
B_v:=
\left\{
i\in[n]:
\langle x_i,v\rangle^2
<
\frac{2\tau}{d}
\right\}
\]
By $(\kappa,\tau)$-goodness, $|B_v|<\frac{\kappa n}{2}$.
Since
$p_i\le1/(\kappa n)$, $\sum_{i\in B_v}p_i
<
\frac12$.
Therefore,
\[
\begin{aligned}
v^\top\Sigma(\mu)v
&=
d\sum_{i=1}^n
p_i\langle x_i,v\rangle^2\\
&\ge
2\tau
\sum_{i\notin B_v}p_i >
\tau
\end{aligned}
\]
Taking the infimum over all unit vectors $v$
gives $\lambda_{\min}(\Sigma(\mu))\ge\tau$.
\end{proof}

\begin{lemma}[Inverse perturbation]
\label{lem:private-inverse-perturbation}
Let $\Sigma\succ0$ and
$\widehat\Sigma=\Sigma+N$, where $N$ is symmetric.
Suppose
\[
\|\Sigma^{-1/2}
N\Sigma^{-1/2}\|_2
\le\eta\le\frac12
\]
Writing $M:=\Sigma^{-1},~
\widehat M:=\widehat\Sigma^{-1}$,
we have
\[
\frac{1}{1+\eta}M
\preceq
\widehat M
\preceq
\frac{1}{1-\eta}M
\]
Consequently, for every nonzero $x$,
\[
\left|
\log(x^\top\widehat Mx)
-
\log(x^\top Mx)
\right|
\le2\eta
\qquad \text{and} \qquad
\left|
\operatorname{tr}(\widehat M\Sigma)-d
\right|
\le2d\eta
\]
\end{lemma}

\begin{proof}
The assumption gives
\[
-\eta I
\preceq
\Sigma^{-1/2}N\Sigma^{-1/2}
\preceq
\eta I
\]
Hence
\[
(1-\eta)\Sigma
\preceq
\widehat\Sigma
\preceq
(1+\eta)\Sigma
\]
Inverting yields
\[
\frac{1}{1+\eta}M
\preceq
\widehat M
\preceq
\frac{1}{1-\eta}M
\]
Therefore,
\[
\frac1{1+\eta}
\le
\frac{x^\top\widehat Mx}
{x^\top Mx}
\le
\frac1{1-\eta}
\]
Since $\eta\le1/2$,
\[
\log(1+\eta)\le\eta,
\qquad
-\log(1-\eta)\le2\eta
\]
which proves the logarithmic-score bound.

Moreover, the eigenvalues of
$\Sigma^{1/2}\widehat M\Sigma^{1/2}$
belong to $\left[
\frac1{1+\eta},
\frac1{1-\eta}
\right]$.
Thus
\[
\left|
\operatorname{tr}(\widehat M\Sigma)-d
\right|
\le
\frac{d\eta}{1-\eta}
\le2d\eta
\]
\end{proof}

\subsection{Stability of the KL-Potential Argument}

\begin{proof}[Proof of Lemma~\ref{lem:private-sum-bound}]
For every round $t$, define
\[
\widehat\phi_i^t
:=
\log\!\left(x_i^\top\widehat M^t x_i\right),
\qquad
\phi_i^t
:=
\phi_i(\mu^t)
=
\log\!\left(x_i^\top M(\mu^t)x_i\right)
\]
and
\[
\theta_i^t
:=
\log\left(\frac{w_i^t}{\mu_i^0}\right),
\qquad
F(\theta)
:=
\max_{\mu\in\Gamma_\kappa}
\left\{
\langle\mu,\theta\rangle
-
D_{\mathrm{KL}}(\mu\Vert\mu^0)
\right\}
\]

Since $w_i^t=\mu_i^0e^{\theta_i^t}$, for every
$\mu\in\Gamma_\kappa$ we have
\[
\langle\mu,\theta^t\rangle
-
D_{\mathrm{KL}}(\mu\Vert\mu^0)
=
-D_{\mathrm{KL}}(\mu\Vert w^t)
+\sum_i w_i^t-d
\]
The last term is independent of $\mu$, and since
$\mu^t=\Pi_{\Gamma_\kappa}(w^t)$, the measure $\mu^t$
attains the maximum in $F(\theta^t)$.

Moreover, by the Pythagorean property of the KL projection
(Theorem~\ref{thm:bregman-pythagorean}), for every
$\nu\in\Gamma_\kappa$,
\[
\begin{aligned}
F(\theta^t)
-
\left(
\langle\nu,\theta^t\rangle
-
D_{\mathrm{KL}}(\nu\Vert\mu^0)
\right)
&=
-D_{\mathrm{KL}}(\mu^t\Vert w^t)
+
D_{\mathrm{KL}}(\nu\Vert w^t)
\ge
D_{\mathrm{KL}}(\nu\Vert\mu^t)
\end{aligned}
\]
Equivalently,
\[
\langle\nu,\theta^t\rangle
-
D_{\mathrm{KL}}(\nu\Vert\mu^0)
\le
F(\theta^t)
-
D_{\mathrm{KL}}(\nu\Vert\mu^t)
\]

The private multiplicative update gives $w_i^{t+1}
=
w_i^t e^{\widehat\phi_i^t}$,
and therefore $\theta_i^{t+1}
=
\theta_i^t+\widehat\phi_i^t$.
Hence,
\begin{align*}
F(\theta^{t+1})
&=
\max_{\mu\in\Gamma_\kappa}
\left\{
\langle\mu,\theta^t\rangle
+
\sum_i\mu_i\widehat\phi_i^t
-
D_{\mathrm{KL}}(\mu\Vert\mu^0)
\right\}
\\
&\le
F(\theta^t)
+
\max_{\mu\in\Gamma_\kappa}
\left\{
\sum_i\mu_i\widehat\phi_i^t
-
D_{\mathrm{KL}}(\mu\Vert\mu^t)
\right\}
\end{align*}
Since both $\mu$ and $\mu^t$ have total mass $d$,
\[
D_{\mathrm{KL}}(\mu\Vert\mu^t)
=
\sum_i\mu_i
\log\left(\frac{\mu_i}{\mu_i^t}\right)
\]
and thus
\begin{align*}
\sum_i\mu_i\widehat\phi_i^t
-
D_{\mathrm{KL}}(\mu\Vert\mu^t)
&=
-\sum_i\mu_i
\log\left(
\frac{\mu_i}
{\mu_i^t e^{\widehat\phi_i^t}}
\right)\stackrel{(\ast)}{\le}
d\log\left(
\frac{
\sum_i\mu_i^t e^{\widehat\phi_i^t}
}{d}
\right)
\end{align*}
where $(\ast)$ follows from the log-sum inequality and
$\sum_i\mu_i=d$. Consequently,
\[
F(\theta^{t+1})
\le
F(\theta^t)
+
d\log\left(
\frac1d
\sum_i\mu_i^t e^{\widehat\phi_i^t}
\right)
\]

By Lemma~\ref{lem:private-inverse-perturbation}, $|\widehat\phi_i^t-\phi_i^t|
\le2\eta$,
and hence $e^{\widehat\phi_i^t}
\le
e^{2\eta}e^{\phi_i^t}$.
Therefore,
\[
\sum_i\mu_i^t e^{\widehat\phi_i^t}
\le
e^{2\eta}
\sum_i\mu_i^t e^{\phi_i^t}
\]
Recall the leverage-score identity, $\sum_i\mu_i^t e^{\phi_i^t}
=
\sum_i
\mu_i^t x_i^\top M(\mu^t)x_i
=
d$.
It follows that
\[
F(\theta^{t+1})
\le
F(\theta^t)
+
d\log(e^{2\eta})
=
F(\theta^t)+2d\eta
\]

Since $\theta^0=0$ and $\mu^0\in\Gamma_\kappa$,
we have $F(\theta^0)=0$. Summing over the $T$ rounds gives $F(\theta^T)
\le
2dT\eta$.

Fix any $q\in\Gamma_\kappa$. By the definition of $F$,
\[
\langle q,\theta^T\rangle
-
D_{\mathrm{KL}}(q\Vert\mu^0)
\le
F(\theta^T)
\]
and since $\theta_i^T
=
\sum_{t=0}^{T-1}\widehat\phi_i^t$,
we obtain
\[
\sum_{t=0}^{T-1}\sum_i
q_i\widehat\phi_i^t
\le
D_{\mathrm{KL}}(q\Vert\mu^0)
+
2dT\eta
\]

Again using
$|\widehat\phi_i^t-\phi_i^t|\le2\eta$
and $\sum_iq_i=d$, we have
\[
\sum_{t=0}^{T-1}\sum_iq_i\phi_i^t
\le
\sum_{t=0}^{T-1}\sum_iq_i\widehat\phi_i^t
+
2dT\eta
\le
D_{\mathrm{KL}}(q\Vert\mu^0)
+
4dT\eta
\]
Dividing by $T$ and applying Jensen's inequality to the convex
functions $\phi_i$ gives
\[
\sum_iq_i\phi_i(\overline\mu)
\le
\frac1T
\sum_{t=0}^{T-1}
\sum_iq_i\phi_i(\mu^t)
\le
\frac{D_{\mathrm{KL}}(q\Vert\mu^0)}{T}
+
4d\eta
\]
as required.
\end{proof}

\subsection{High-Probability Utility Guarantee}

\begin{lemma}[Simultaneous control of the Gaussian perturbations]
\label{lem:private-noise-event}
Suppose that each oracle call uses symmetric
Gaussian noise with scale $\sigma
=
\frac{2dR^2}
{\kappa n\sqrt{2\rho/(T+1)}}$.
If
\[
n
\ge
\frac{
32dR^2
\left(
\sqrt d+
\sqrt{2\log((T+1)/\beta)}
\right)
}{
\kappa\gamma\tau
\sqrt{2\rho/(T+1)}
}
\]
then, with probability at least $1-\beta$,
\[
\max_{0\le t\le T}
\|N^t\|_2
\le
\frac{\gamma\tau}{8}
\]
\end{lemma}

\begin{proof}
For each symmetric Gaussian perturbation,
the standard spectral-norm concentration bound
gives
\[
\Pr\left[
\|N^t\|_2
>
\sigma
\left(
\sqrt d+
\sqrt{2\log((T+1)/\beta)}
\right)
\right]
\le
\frac{\beta}{T+1}
\]
A union bound over the $T+1$ oracle calls gives,
with probability at least $1-\beta$,
\[
\max_{0\le t\le T}
\|N^t\|_2
\le
\sigma
\left(
\sqrt d+
\sqrt{2\log((T+1)/\beta)}
\right)
\]
The assumed lower bound on $n$ makes the
right-hand side at most $\gamma\tau/8$.
\end{proof}

\begin{lemma}[Utility of the averaged measure]
\label{lem:private-averaged-utility}
Under the assumptions of
Theorem~\ref{thm:main-private}, on the event
of Lemma~\ref{lem:private-noise-event},
\[
\sup_{q\in\Gamma_\kappa}
\sum_{i=1}^n
q_i
\log\!\left(
x_i^\top M(\overline\mu)x_i
\right)
\le d\gamma
\]
Consequently, there exists
$S\subseteq[n]$ with
$|S|\ge(1-\kappa)n$ such that
\[
e^{-\gamma/2}
E(\overline\mu)
\subseteq K_S
\]
and
\[
\operatorname{vol}
\left(
e^{-\gamma/2}E(\overline\mu)
\right)
\ge
e^{-d\gamma/2}
\operatorname{vol}(E^*)
\]
\end{lemma}

\begin{proof}
Set $\eta=\gamma/8$.
By Lemma~\ref{lem:private-sum-bound},
for every $q\in\Gamma_\kappa$,
\[
\begin{aligned}
\sum_iq_i
\log\!\left(
x_i^\top M(\overline\mu)x_i
\right)
&\le
\frac{
D_{\mathrm{KL}}(q\Vert\mu^0)
}{T}
+
4d\eta\\
&\le
\frac{
d\log(1/\kappa)
}{T}
+
\frac{d\gamma}{2}
\end{aligned}
\]
Since
$T\ge2\log(1/\kappa)/\gamma$,
the right-hand side is at most $d\gamma$.

The trimmed-containment and volume guarantees
then follow exactly as in the non-private case.
\end{proof}

\subsection{Utility of the Final Released Ellipsoid}

\begin{lemma}
\label{lem:private-final-output}
Let
\[
M:=M(\overline\mu),
\qquad
\widehat M
:=
\left(
\Sigma(\overline\mu)+N^T
\right)^{-1}
\]
On the event of
Lemma~\ref{lem:private-noise-event},
\[
\frac{1}{1+\eta}M
\preceq
\widehat M
\preceq
\frac{1}{1-\eta}M,
\qquad
\eta=\frac{\gamma}{8}
\]
Consequently, the ellipsoid
\[
\widehat E
:=
\left\{
x\in\mathbb R^d:
x^\top
\left(
(1-\eta)e^{-\gamma}\widehat M
\right)^{-1}
x
\le1
\right\}
\]
satisfies $\widehat E\subseteq K_S$
for some $S\subseteq[n]$ with
$|S|\ge(1-\kappa)n$, and
\[
\operatorname{vol}(\widehat E)
\ge
e^{-d\gamma/2}
\left(
\frac{1-\eta}{1+\eta}
\right)^{d/2}
\operatorname{vol}(E^*)
\]
Moreover, $K_P
\subseteq
\sqrt{d(1+\eta)}\,\widehat E_0$,
where $\widehat E_0
=
\{x:x^\top\widehat M^{-1}x\le1\}$.
\end{lemma}

\begin{proof}
The matrix inequalities follow from
Lemma~\ref{lem:private-inverse-perturbation}.

Let
\[
M_\gamma:=e^{-\gamma}M,
\qquad
\widehat M_\gamma
:=
(1-\eta)e^{-\gamma}\widehat M
\]
Since $\widehat M
\preceq
\frac{1}{1-\eta}M$,
we have $\widehat M_\gamma
\preceq
M_\gamma$.
Therefore, the ellipsoid associated with
$\widehat M_\gamma$ is contained in the
ellipsoid associated with $M_\gamma$.
By Lemma~\ref{lem:private-averaged-utility},
the latter is contained in $K_S$, and hence $\widehat E\subseteq K_S$.

For the volume guarantee, $\widehat M
\succeq
\frac{1}{1+\eta}M$,
so $\widehat M_\gamma
\succeq
\frac{1-\eta}{1+\eta}
M_\gamma$.
Taking determinants gives
\[
\operatorname{vol}(\widehat E)
\ge
\left(
\frac{1-\eta}{1+\eta}
\right)^{d/2}
\operatorname{vol}
\left(
e^{-\gamma/2}E(\overline\mu)
\right)
\]
and the desired bound follows from
Lemma~\ref{lem:private-averaged-utility}.

Finally, for every $x\in K_P$, $x^\top M^{-1}x\le d$.
Since
$\widehat M\succeq M/(1+\eta)$,
we obtain
\[
\widehat M^{-1}
\preceq
(1+\eta)M^{-1}
\]
and therefore
\[
x^\top\widehat M^{-1}x
\le
d(1+\eta)
\]
Thus $K_P
\subseteq
\sqrt{d(1+\eta)}\,\widehat E_0$.
\end{proof}

\begin{proof}[Proof of Theorem~\ref{thm:main-private}]
By Lemma~\ref{lem:private-noise-event}, with
probability at least $1-\beta$ all $T+1$
perturbations satisfy
\[
\|N^t\|_2
\le
\eta\tau,
\qquad
\eta=\gamma/8
\]
Lemma~\ref{lem:good-assumption} implies
$\lambda_{\min}(\Sigma(\mu^t))\ge\tau$,
and hence
\[
\lambda_{\min}
(\Sigma(\mu^t)+N^t)
\ge
(1-\eta)\tau
\ge
\tau/2
\]
Thus no oracle invocation returns $\perp$.

Lemma~\ref{lem:private-averaged-utility}
establishes the guarantee for the averaged
measure, and
Lemma~\ref{lem:private-final-output}
transfers it to the released matrix
$\widehat M$.
\end{proof}

\section{Population Characterization of The Goodness Parameter}\label{apx:population}

The $(\kappa,\tau)$-goodness condition used in our private analysis is a finite-sample property of the input dataset. In this section, we show that this condition admits a natural population-level characterization for every distribution $\mathcal P$ over $\mathbb R^d$. In particular, for any $\kappa\in(0,1)$, we define a canonical distribution-dependent parameter $\tau_{\mathcal P,\kappa}$ in terms of the lower quantiles of the one-dimensional squared projections $\langle X,v\rangle^2$. We then show that, with sufficiently many i.i.d.\ samples from $\mathcal P$, the resulting dataset is $(\kappa,\tau_{\mathcal P,\kappa})$-good with high probability. Thus, rather than viewing $\tau$ as an arbitrary structural parameter of a finite dataset, it can be interpreted as a population quantity measuring how far the distribution is uniformly bounded away from degeneracy across all directions. The proof follows from uniform convergence over the class of centered slabs, whose VC dimension is $O(d)$.

\subsection{Preliminaries for this section: VC dimension}

\begin{definition}[VC dimension]
Let $\mathcal{A}$ be a class of subsets of a domain $\mathcal{X}$.
We say that a finite set $S=\{x_1,\ldots,x_m\}\subseteq\mathcal{X}$ is
\emph{shattered} by $\mathcal{A}$ if, for every subset $T\subseteq S$,
there exists a set $A\in\mathcal{A}$ such that
\[
A\cap S = T
\]
The \emph{Vapnik--Chervonenkis (VC) dimension} of $\mathcal{A}$,
denoted by $\operatorname{VCdim}(\mathcal{A})$, is the largest integer
$m$ for which there exists a set $S\subseteq\mathcal{X}$ of cardinality
$m$ that is shattered by $\mathcal{A}$. If arbitrarily large finite sets
can be shattered by $\mathcal{A}$, then
$\operatorname{VCdim}(\mathcal{A})=\infty$.
\end{definition}

\begin{theorem}[Uniform convergence for VC classes~\citep{Vapnik2015}]
\label{thm:vc-uniform-convergence}
Let $\mathcal{A}$ be a class of measurable subsets of a domain $\mathcal{X}$
with VC dimension
\[
\operatorname{VCdim}(\mathcal{A}) = v
\]
Let $X_1,\ldots,X_n \overset{\mathrm{i.i.d.}}{\sim} \mathcal{P}$.
Then there exists a universal constant $C>0$ such that, for every
$\epsilon,\beta\in(0,1)$, if
\[
n
\ge
C
\frac{
v\log(1/\epsilon)+\log(1/\beta)
}{
\epsilon^2
}
\]
then, with probability at least $1-\beta$,
\[
\sup_{A\in\mathcal{A}}
\left|
\frac{1}{n}\sum_{i=1}^n \mathbf{1}\{X_i\in A\}
-
\Pr_{X\sim\mathcal{P}}[X\in A]
\right|
\le \epsilon
\]
\end{theorem}

\begin{lemma}[Sauer's Lemma]
\label{lem:sauer}
Let $\mathcal H$ be a class of subsets of a domain $\mathcal X$ with
VC dimension $\nu<\infty$. Let $\Pi_{\mathcal H}(m)$ denote the growth
function of $\mathcal H$, defined by
\[
\Pi_{\mathcal H}(m)
:=
\max_{\substack{S\subseteq\mathcal X\\ |S|=m}}
\left|
\left\{
H\cap S:H\in\mathcal H
\right\}
\right|
\]
Then, for every integer $m\ge 0$,
\[
\Pi_{\mathcal H}(m)
\le
\sum_{j=0}^{\nu}
\binom{m}{j}
\]
In particular, if $m\ge \nu\ge 1$, then
\[
\Pi_{\mathcal H}(m)
\le
\sum_{j=0}^{\nu}
\binom{m}{j}
\le
\left(\frac{em}{\nu}\right)^\nu
\]
\end{lemma}

\subsection{Population Characterization of The Goodness
Parameter}
\begin{lemma}
\label{lem:population-goodness}
Let $\mathcal P$ be a distribution over $\mathbb R^d$, let
$\kappa,\beta\in(0,1)$. For every unit vector $v\in\mathbb S^{d-1}$, define the cumulative
distribution function
\[
F_v(t)
:=
\Pr_{x\sim\mathcal P}
\left[
\langle x,v\rangle^2\le t
\right]
\]
and its generalized inverse
\[
F_v^{-1}(p)
:=
\inf\left\{
t\ge0:F_v(t)\ge p
\right\}
\]
Define
\[
\tau_{\mathcal P,\kappa}
:=
\frac d2
\inf_{\|v\|_2=1}
F_v^{-1}\left(\frac{\kappa}{4}\right)
\]

Then there exists a universal constant $C>0$ such that, if
\[
n
\ge
C
\frac{
d\log(1/\kappa)+\log(1/\beta)
}{
\kappa^2
}
\]
then, with probability at least $1-\beta$, the sample $P=\{x_1,\ldots,x_n\} \sim \mathcal P^n$
is $(\kappa,\tau_{\mathcal P,\kappa})$-good. That is, simultaneously
for every unit vector $v\in\mathbb S^{d-1}$,
\[
\left|
\left\{
i\in[n]:
\langle x_i,v\rangle^2
<
\frac{2\tau_{\mathcal P,\kappa}}{d}
\right\}
\right|
<
\frac{\kappa n}{2}
\]
\end{lemma}

\begin{proof}
Set
\[
a
:=
\frac{2\tau_{\mathcal P,\kappa}}{d}
=
\inf_{\|v\|_2=1}
F_v^{-1}\left(\frac{\kappa}{4}\right)
\]
For every unit vector $v$, define
\[
A_v
:=
\left\{
x\in\mathbb R^d:
\langle x,v\rangle^2<a
\right\}
\]

By the definition of $a$,
\[
a
\le
F_v^{-1}\left(\frac{\kappa}{4}\right)
\qquad
\text{for every }v\in\mathbb S^{d-1}
\]
Hence
\[
\Pr_{x\sim\mathcal P}[x\in A_v]
=
\Pr_{x\sim\mathcal P}
\left[
\langle x,v\rangle^2<a
\right]
\le
\frac{\kappa}{4}
\]
Indeed, by the definition of the generalized inverse,
\[
\Pr\left[
\langle x,v\rangle^2
<
F_v^{-1}\left(\frac{\kappa}{4}\right)
\right]
\le
\frac{\kappa}{4}
\]

Now consider the class of sets
\[
\mathcal A
:=
\left\{
A_v:v\in\mathbb S^{d-1}
\right\}
\]
Each $A_v$ is a centered slab of the form
\[
A_v
=
\left\{
x:
-\sqrt a<\langle x,v\rangle<\sqrt a
\right\}
\]
We now use the following lemma.

\begin{lemma}[VC dimension of centered slabs]
\label{lem:vc-centered-slabs}
Fix $a>0$, and consider the class
\[
\mathcal A
=
\left\{
A_v
=
\left\{
x\in\mathbb R^d:
-\sqrt a<\langle x,v\rangle<\sqrt a
\right\}
:
v\in\mathbb S^{d-1}
\right\}
\]
Then
\[
\operatorname{VCdim}(\mathcal A)=O(d)
\]
More precisely,
\[
\operatorname{VCdim}(\mathcal A)<16(d+1)
\]
\end{lemma}

\begin{proof}
Let $\mathcal H_d$ denote the class of affine halfspaces in
$\mathbb R^d$. It is well known that
\[
\operatorname{VCdim}(\mathcal H_d)=d+1
\]

For every $v\in\mathbb S^{d-1}$, we can write
\[
A_v
=
\left\{
x:\langle x,v\rangle<\sqrt a
\right\}
\cap
\left\{
x:-\langle x,v\rangle<\sqrt a
\right\}
\]
Thus every set in $\mathcal A$ is the intersection of two affine
halfspaces. Therefore,
\[
\mathcal A
\subseteq
\mathcal H_d^{\cap 2}
\]
where
\[
\mathcal H_d^{\cap 2}
:=
\left\{
H_1\cap H_2:
H_1,H_2\in\mathcal H_d
\right\}
\]

Let $\Pi_{\mathcal C}(m)$ denote the growth function of a class
$\mathcal C$, namely the maximum number of distinct subsets that
$\mathcal C$ can induce on a set of $m$ points.

For any $m$-point set $S$, each set of the form
\[
(H_1\cap H_2)\cap S
\]
is determined by the pair
\[
(H_1\cap S,H_2\cap S)
\]
Hence
\[
\Pi_{\mathcal A}(m)
\le
\Pi_{\mathcal H_d^{\cap 2}}(m)
\le
\Pi_{\mathcal H_d}(m)^2
\]

Since $\operatorname{VCdim}(\mathcal H_d)=d+1$, Sauer's lemma~(Lemma \ref{lem:sauer}) gives,
for $m\ge d+1$,
\[
\Pi_{\mathcal H_d}(m)
\le
\left(
\frac{em}{d+1}
\right)^{d+1}
\]
Therefore,
\[
\Pi_{\mathcal A}(m)
\le
\left(
\frac{em}{d+1}
\right)^{2(d+1)}
\]

Suppose that $\mathcal A$ shatters a set of size $m$. Then
\[
2^m
\le
\Pi_{\mathcal A}(m)
\]
and consequently
\[
2^m
\le
\left(
\frac{em}{d+1}
\right)^{2(d+1)}
\]

Now take $m=16(d+1)$. Then the above inequality would imply
\[
2^{16(d+1)}
\le
(16e)^{2(d+1)}
\]
However, $2^{16}>(16e)^2$, and therefore $2^{16(d+1)} > (16e)^{2(d+1)}$. This is a contradiction.

Hence no set of size $16(d+1)$ can be shattered by $\mathcal A$, so
\[
\operatorname{VCdim}(\mathcal A)
<
16(d+1)
\]
In particular,
\[
\operatorname{VCdim}(\mathcal A)=O(d)
\]
\end{proof}

By Lemma~\ref{lem:vc-centered-slabs}, $\operatorname{VCdim}(\mathcal A)=O(d)$. Therefore, by setting $\nu = O(d)$ and $\epsilon = \kappa/4$ in Theorem~\ref{thm:vc-uniform-convergence}, there exists a
universal constant $C>0$ such that, if
\[
n
\ge
C
\frac{
d\log(1/\kappa)+\log(1/\beta)
}{
\kappa^2
}
\]
then, with probability at least $1-\beta$,
\[
\sup_{A\in\mathcal A}
\left|
\frac1n\sum_{i=1}^n\mathbf 1\{x_i\in A\}
-
\Pr_{x\sim\mathcal P}[x\in A]
\right|
<
\frac{\kappa}{4}
\]

Condition on this event. Then, simultaneously for every unit vector $v$,
\begin{align*}
\frac1n
\left|
\left\{
i\in[n]:
\langle x_i,v\rangle^2
<
\frac{2\tau_{\mathcal P,\kappa}}{d}
\right\}
\right|
&=
\frac1n
\sum_{i=1}^n
\mathbf 1\{x_i\in A_v\}
\\
&<
\Pr_{x\sim\mathcal P}[x\in A_v]
+
\frac{\kappa}{4}
\le
\frac{\kappa}{4}
+
\frac{\kappa}{4}
=
\frac{\kappa}{2}.
\end{align*}
Thus,
\[
\left|
\left\{
i\in[n]:
\langle x_i,v\rangle^2
<
\frac{2\tau_{\mathcal P,\kappa}}{d}
\right\}
\right|
<
\frac{\kappa n}{2}
\]
for every unit vector $v$, and hence $P$ is
$(\kappa,\tau_{\mathcal P,\kappa})$-good.
\end{proof}

\section{Reduction from uncentered to centered trimmed MVEE}
\label{apx:uncentered-centered-reduction}
\begin{lemma}
\label{lem:uncentered-centered-reduction}
Let $Y=\{y_1,\ldots,y_n\}\subseteq \mathbb{R}^d$
and define the lifted points
\[
x_i :=
\begin{pmatrix}
y_i\\
1
\end{pmatrix}
\in\mathbb{R}^{d+1},
\qquad i\in[n]
\]
Let \(X:=\{x_1,\ldots,x_n\}\).

Suppose that a centered MVEE algorithm applied to \(X\) outputs the ellipsoid
\[
\widehat E_X
=
\left\{
x\in\mathbb{R}^{d+1}:
x^\top \widehat M x\le 1
\right\},
\qquad
\widehat M\succ0
\]
which contains at least a \((1-\kappa)\)-fraction of the lifted points, denoted $S^\star$, and
whose volume is at most \(e^\gamma\) times that of the optimal centered
ellipsoid containing $X$.

Write
\[
\widehat M
=
\begin{pmatrix}
A & b\\
b^\top & c
\end{pmatrix}
\]
and define
\[
\widehat y:=-A^{-1}b,
\qquad
s:=c-b^\top A^{-1}b
\]
Then
\[
\widehat E_Y
:=
\left\{
y\in\mathbb{R}^d:
(y-\widehat y)^\top
\widehat M_Y
(y-\widehat y)
\le 1
\right\},
\qquad
\widehat M_Y:=\frac{A}{1-s}
\]
contains $S^\star$, the same subset of original points as \(\widehat E_X\) contains
among the lifted points. In particular,
\[
\left|
\left\{
i:y_i\in\widehat E_Y
\right\}
\right|
\ge (1-\kappa)n
\]
Moreover,
\[
\text{vol}_d(\widehat E_Y)
\le
e^\gamma
\text{OPT}^{\mathrm{unc}}(Y)
\]
where \(\text{OPT}^{\mathrm{unc}}(Y)\) is the minimum volume of an
uncentered ellipsoid containing \(Y\).
\end{lemma}

\begin{proof}
For every \(y\in\mathbb{R}^d\),
\begin{align*}
\begin{pmatrix}
y\\
1
\end{pmatrix}^{\!\top}
\widehat M
\begin{pmatrix}
y\\
1
\end{pmatrix}
&=
y^\top Ay+2b^\top y+c\\
&=
(y-\widehat y)^\top A(y-\widehat y)
+
c-b^\top A^{-1}b\\
&=
(y-\widehat y)^\top A(y-\widehat y)+s
\end{align*}
Hence
\[
\begin{pmatrix}
y\\
1
\end{pmatrix}
\in \widehat E_X
\quad\Longleftrightarrow\quad
(y-\widehat y)^\top A(y-\widehat y)\le 1-s
\]
Since \(\widehat M\succ0\), its Schur complement satisfies \(s>0\);
moreover, since \(\widehat E_X\) contains at least one lifted point,
we have \(s<1\). Therefore
\[
\widehat M_Y=\frac{A}{1-s}
\]
is well defined and
\[
\begin{pmatrix}
y\\
1
\end{pmatrix}
\in \widehat E_X
\quad\Longleftrightarrow\quad
(y-\widehat y)^\top
\widehat M_Y
(y-\widehat y)
\le1
\]
Consequently, for every \(i\in[n]\),
\[
x_i\in\widehat E_X
\quad\Longleftrightarrow\quad
y_i\in\widehat E_Y
\]
and thus
\[
\left|
\left\{
i:y_i\in\widehat E_Y
\right\}
\right|
=
\left|
\left\{
i:x_i\in\widehat E_X
\right\}
\right|
\ge(1-\kappa)n
\]

It remains to establish the approximation guarantee. By the Schur
complement identity,
\[
\det(\widehat M)=\det(A)\cdot s
\]
whereas
\[
\det(\widehat M_Y)
=
\frac{\det(A)}{(1-s)^d}
\]
Therefore
\[
\frac{\det(\widehat M)}
{\det(\widehat M_Y)}
=
s(1-s)^d
\]
The function
\[
f(s)=s(1-s)^d,
\qquad s\in(0,1)
\]
is maximized at $s=\frac{1}{d+1}$,
and hence
\[
s(1-s)^d
\le
\frac{d^d}{(d+1)^{d+1}}
\]
Denoting $C_d:=\frac{d^d}{(d+1)^{d+1}}$
we obtain
\[
\det(\widehat M_Y)
\ge
\frac{\det(\widehat M)}{C_d}
\tag{1}
\]

Now fix any subset \(S\subseteq[n]\). Let \(M_S^\star\) be the matrix
of the minimum-volume centered ellipsoid of the lifted points
\(\{x_i:i\in S\}\), written as
\[
\{x:x^\top M_S^\star x\le1\}
\]
and let \(M_{Y,S}^\star\) be the matrix of the minimum-volume
uncentered ellipsoid of the original points
\(\{y_i:i\in S\}\), written as
\[
\{y:(y-y_S^\star)^\top M_{Y,S}^\star
(y-y_S^\star)\le1\}
\]
By the standard lifting reduction from uncentered to centered MVEE,
the Schur complement parameter of the optimal lifted solution is
\[
s^\star=\frac{1}{d+1}
\]
and its leading block satisfies
\[
A^\star=\frac{d}{d+1}M_{Y,S}^\star
\]
Consequently,
\[
\det(M_S^\star)
=
\frac{d^d}{(d+1)^{d+1}}
\det(M_{Y,S}^\star)
=
C_d\det(M_{Y,S}^\star)
\tag{2}
\]

Since the relation \((2)\) holds for every subset \(S\), it also holds
for $S^\star$. Thus, if \(M^\star\) and
\(M_{Y}^\star\) denote optimal matrices for the lifted centered
and original uncentered trimmed problems, respectively, then
\[
\det(M^\star)
=
C_d\det(M_{Y}^\star)
\tag{3}
\]

Using \((1)\) and \((3)\), we obtain
\begin{align*}
\frac{\text{vol}_d(\widehat E_Y)}
{\text{OPT}^{\mathrm{unc}}(Y)}
&\stackrel{(\ast)}=
\sqrt{
\frac{\det(M_{Y}^\star)}
{\det(\widehat M_Y)}
}\\
&\le
\sqrt{
\frac{\det(M^\star)/C_d}
{\det(\widehat M)/C_d}
}\\
&=
\sqrt{
\frac{\det(M^\star)}
{\det(\widehat M)}
} \stackrel{(\ast)} = \frac{\text{vol}_{d+1}(\widehat E_X)}
{\text{OPT}^{\mathrm{cen}}(X)}
\le e^\gamma
\end{align*}
where the inequalities marked with $(\ast)$ follow from the fact that an ellipsoid $\{x: ~ x^T M x \leq 1\}$ has volume proportional to
\(\det(M)^{-1/2}\).
Therefore
\[
\text{vol}_d(\widehat E_Y)
\le
e^\gamma
\text{OPT}^{\mathrm{unc}}(Y)
\]
as claimed.
\end{proof}

\end{document}